\documentclass{article}
\usepackage[preprint]{neurips_2025}

\usepackage[utf8]{inputenc} % allow utf-8 input
\usepackage[T1]{fontenc}    % use 8-bit T1 fonts
\usepackage{hyperref}       % hyperlinks
\usepackage{url}            % simple URL typesetting
\usepackage{booktabs}       % professional-quality tables
\usepackage{amsfonts}       % blackboard math symbols
\usepackage{nicefrac}       % compact symbols for 1/2, etc.
\usepackage{microtype}      % microtypography
\usepackage{xcolor}         % colors
\usepackage{multirow}

\usepackage{enumitem}
\usepackage{amsmath, amsthm}
\usepackage[capitalise]{cleveref}
\usepackage{graphicx}
\usepackage{wrapfig} 
\newtheorem{theorem}{Theorem}
\newtheorem{lemma}{Lemma}
\newtheorem{corollary}{Corollary}
\newtheorem{proposition}{Proposition}
\newtheorem{assumption}{Assumption}
\usepackage{algorithm}
\usepackage{algorithmic}
\usepackage{subcaption}
\usepackage{siunitx}
\usepackage{tabularx}

\usepackage{amssymb,amsfonts}
\usepackage{bm}
\usepackage{mathtools}

\newcommand{\E}{\mathbb{E}}
\newcommand{\R}{\mathbb{R}}
\newcommand{\F}{\mathrm{F}}
\newcommand{\norm}[1]{\left\lVert #1 \right\rVert}
\newcommand{\abs}[1]{\left\lvert #1 \right\rvert}
\newcommand{\dd}{\,\mathrm{d}}
\newcommand{\defeq}{\triangleq}
\newcommand{\Var}{\mathrm{Var}}

\newcommand{\logtwo}{\log_2}

\title{Optimal Scalar Quantization for Matrix Multiplication: Closed-Form Density and Phase Transition}

\author{%
  Calvin Ang\thanks{Equal contribution.} \\
  Stanford University \\
  \And
  Sungyoon Kim\footnotemark[1] \\
  Stanford University \\
  \And
  Mert Pilanci \\
  Stanford University \\
}

\begin{document}

\maketitle

\begin{abstract}
We study entrywise scalar quantization of two matrices prior to multiplication. Given $A\in\R^{m\times k}$ and $B\in\R^{k\times n}$, we quantize entries of $A$ and $B$ independently using scalar quantizers with $K_X$ and $K_Y$ levels per entry, and form $\widehat C=\widehat A\,\widehat B$. The objective is to minimize the matrix multiplication mean-squared error (MSE) $\mathcal{E}\defeq \E[\|{AB-\widehat A\widehat B}\|_F^2]$ under a pair-i.i.d.\ inner-product model.
In the high-resolution regime $K_X,K_Y\to\infty$, we derive a sharp $K^{-2}$ asymptotic expansion for $\mathcal{E}$, identify the exact optimal leading constants, and characterize asymptotically optimal quantization center densities in terms of conditional second moments. We then specialize to correlated Gaussian multiplicative pairs, obtaining a closed-form optimal point density
\[
\lambda^\star(u)\ \propto\ \exp\!\left(-\frac{u^2}{6}\right)\bigl((1-\rho^2)+\rho^2u^2\bigr)^{1/3},
\qquad u=\frac{x}{\sigma_X},
\]
with the same form for $y/\sigma_Y$, and prove a correlation-driven phase transition: the density is unimodal at the origin for $\abs{\rho}\le 1/\sqrt{3}$ and becomes bimodal for $\abs{\rho}>1/\sqrt{3}$ with peaks at $u_{\mathrm{peak}}=\pm\sqrt{3-1/\rho^2}$. We show our method's applicability in synthetic experiments such as matrix multiplication quantization and least squares optimization, as well as quantization of large language model key and query activations.
% Finally, we add information-theoretic (rate--distortion) lower bounds for product reconstruction and inner products, and compare their $2^{-2R}$ scaling to the operational high-rate upper bounds for factor-quantize-and-multiply schemes, clarifying the fundamental gap between ``coding the product'' and ``coding the factors''.
\end{abstract}

\section{Introduction}
Quantized matrix multiplication is central in modern machine learning inference and scientific computing. In many deployments, one does \emph{not} care about entrywise reconstruction error of operands $A$ and $B$; rather, the relevant distortion is the error induced \emph{after multiplication} \cite{ordentlichpolyanskiy2025optimal}. This perspective is closely connected to classical high-rate theory \cite{bennett1948spectra} and to the broader quantization literature \cite{gersho2012vector}, which show that optimal quantizer design depends fundamentally on the downstream distortion criterion. At the same time, low-precision quantization has become a practical necessity in modern GPU inference: recent hardware and software stacks increasingly rely on ultra-low-precision formats such as NVFP4 \cite{nvidia2025nvfp4}, INT8 \cite{dettmers2022llmint88bitmatrixmultiplication}, and FP8 \cite{micikevicius2022fp8formatsdeeplearning} to improve throughput, memory efficiency, and energy efficiency for large-scale matrix multiplication workloads arising in large language model inference . These trends make it especially important to understand quantization schemes that are optimal for the \emph{product} itself, rather than for the separate reconstruction of the input matrices. However, matrix multiplication introduces a distinct bilinear distortion structure: the error contributed by one operand is filtered by the other, and the importance of each entry depends jointly on both factors. Related recent work has studied quantized matrix multiplication in different settings, including nested lattice quantization \cite{ordentlichpolyanskiy2025optimal}. In contrast, our focus is on the optimal design of scalar quantizers specifically for matrix multiplication, deriving high-rate laws tailored to product distortion and showing how multiplicative structure and statistical dependence fundamentally change the optimal point density.

This paper asks:
\begin{quote}
\emph{Which scalar quantizers for entries of $A$ and $B$ minimize the expected Frobenius MSE of the product computed from quantized operands?}
\end{quote}

\paragraph*{Contributions.}
\begin{itemize}
\item We derive a high-rate characterization of the optimal achievable matrix multiplication MSE, proving a sharp $K^{-2}$ scaling law with an \emph{exact} leading constant.
\item We reduce the matrix objective to two \emph{weighted scalar} MSE quantization problems driven by conditional second moments, and we identify the optimal companding point densities.
\item For correlated Gaussian multiplicative pairs, we obtain closed-form optimal densities and prove a correlation-induced unimodal-to-bimodal ``phase transition''.
% \item We add information-theoretic lower bounds (Shannon lower bound / entropy power benchmarks) for product and inner-product reconstruction, and compare them with our operational upper bounds.
\item Provide experimental validation of our results on synthetically generated matrices for matrix multiplication and quantized least squares as well as on key and query activations on the GPT-2 and Qwen3 family of models.
\end{itemize}

\section{High-Rate Analysis of Matrix Multiplication MSE}\label{sec:highrate_from_scratch}

This section derives the high-rate reduction from matrix multiplication MSE to two decoupled weighted scalar criteria, and states the resulting sharp $K^{-2}$ constant. All asymptotic steps are justified in Appendix~\ref{app:proof_decoupling} and Appendix~\ref{app:proof_scalar_theorem}.

\subsection{Problem Formulation}

Let $A\in\R^{m\times k}$, $B\in\R^{k\times n}$, and $C=AB$. Let $Q_X$ and $Q_Y$ be scalar quantizers with $\abs{\mathrm{range}(Q_X)}\le K_X$ and $\abs{\mathrm{range}(Q_Y)}\le K_Y$. Define
\[
\widehat A_{i\ell}=Q_X(A_{i\ell}),\qquad
\widehat B_{\ell j}=Q_Y(B_{\ell j}),\qquad
\widehat C=\widehat A\,\widehat B.
\]
We measure performance by the matrix multiplication MSE
\[
\mathcal{E}(Q_X,Q_Y)\defeq \E\big[\norm{C-\widehat C}_\F^2\big].
\]

The quantizer design is governed by the distribution of multiplicative entry-pairs feeding each inner product. We assume each pairs $(X, Y)$ in the matrix multiplication has identical joint distributions.

\begin{assumption}[Pair-i.i.d.\ inner products]\label{as:pair_iid}
For each output entry $C_{ij}=\sum_{\ell=1}^k A_{i\ell}B_{\ell j}$, the pairs $\{(A_{i\ell},B_{\ell j})\}_{\ell=1}^k$ are i.i.d.\ copies of a generic pair $(X,Y)$ with finite moments up to order $4+\epsilon$ for some $\epsilon>0$. Because of this, the entries $\{C_{ij}\}_{i\in [m], j \in [n]}$ are identically distributed (although dependent).
\end{assumption}

Under Assumption~\ref{as:pair_iid}, when we denote $(X, Y) \sim \mathcal{D}$ we have 
\[
\mathcal{E}(Q_X, Q_Y) = nm\ \mathbb{E}_{(X_i, Y_i)\sim \mathcal{D}} \Big[\big(\sum_{i=1}^{k}X_iY_i - \sum_{i=1}^{k} \hat{X}_i\hat{Y}_i\big)^2\Big].
\]
We use the standard fixed-rate identification
\[
R_X\defeq \logtwo K_X,\qquad R_Y\defeq \logtwo K_Y,
\]
and analyze the high-resolution regime $K_X,K_Y\to\infty$ (or equivalently, $R_X,R_Y\to\infty$).

\subsection{Sharp High-rate Constant and Optimal Point Densities}

Let $D\defeq XY-\widehat X\widehat Y$ for a single pair $(X,Y)$. Independence across samples $(X_i, Y_i)$ yields
\begin{equation}\label{eq:sum_variance_intro}
\E_{(X_i, Y_i)\sim \mathcal{D}} \big[(\sum_{i=1}^{k}X_iY_i - \sum_{i=1}^{k} \hat{X}_i\hat{Y}_i)^2\big]
= k\,\E[D^2] + k(k-1)\E[D]^2.
\end{equation}
Appendix~\ref{app:proof_decoupling} proves that, for the companding quantizers used in the high-rate design, $\E[D]=O(K_X^{-2}+K_Y^{-2})$. Consequently the second term in \eqref{eq:sum_variance_intro} is $o(K_X^{-2}+K_Y^{-2})$ and does not affect the $K^{-2}$ leading constant.

Define quantization errors $e_X\defeq X-\widehat X$ and $e_Y\defeq Y-\widehat Y$. Then
\begin{equation}\label{eq:bilinear_error_intro}
D=XY-\widehat X\widehat Y
=Y e_X + X e_Y - e_X e_Y.
\end{equation}
Squaring produces
\begin{align}
D^2
&=Y^2e_X^2 + X^2e_Y^2 + e_X^2e_Y^2
+2XYe_Xe_Y -2Ye_X^2e_Y -2Xe_Xe_Y^2. \label{eq:D2_expand_intro}
\end{align}
Appendix~\ref{app:proof_decoupling} shows that the mixed terms on the second line have expectation $O(K_X^{-2}K_Y^{-2})$ for the companding quantizers used in the high-rate design. Since $K_X^{-2}K_Y^{-2}=o(K_X^{-2}+K_Y^{-2})$ as $K_X,K_Y\to\infty$, we obtain the sharp leading reduction
\begin{equation}\label{eq:leading_D2_intro}
\E[D^2] = \E[Y^2e_X^2]+\E[X^2e_Y^2] + o(K_X^{-2}+K_Y^{-2}).
\end{equation}

Because $e_X$ is a deterministic function of $X$ and $e_Y$ is a deterministic function of $Y$, conditioning yields
\[
\E[Y^2e_X^2]=\E\!\left[\E[Y^2\mid X]\ e_X^2\right],\qquad
\E[X^2e_Y^2]=\E\!\left[\E[X^2\mid Y]\ e_Y^2\right].
\]
Define the conditional second moments as
\begin{equation}\label{eq:weights_def_intro}
w_X(x)\defeq \E[Y^2\mid X=x],\qquad
w_Y(y)\defeq \E[X^2\mid Y=y].
\end{equation}
Then \eqref{eq:sum_variance_intro} and \eqref{eq:leading_D2_intro}, along with the fact that $\E[D] = O(K_X^{-2} + K_Y^{-2})$ imply the two-scalar high-rate reduction
\begin{align}\label{eq:two_scalar_intro}
\mathcal{E}(Q_X, Q_Y)
&=mnk\Big(\E[w_X(X)(X-Q_X(X))^2]+\E[w_Y(Y)(Y-Q_Y(Y))^2]\Big)\nonumber\\
&\quad + o(K_X^{-2}+K_Y^{-2}).
\end{align}

Now let's denote $f_X$ and $f_Y$ the marginal densities of $X$ and $Y$, and further assume a regularity condition. These conditions are used to justify the asymptotic analysis and to derive weighted optimal rates.

\begin{assumption}[High-rate regularity]\label{as:regularity}
The pair $(X,Y)$ admits a joint density $f_{X,Y}$ that is twice continuously differentiable.
The functions $f_X,f_Y,w_X,w_Y$ are continuous and strictly positive on $\R$, and
\[
\int_\R (f_X(x)w_X(x))^{1/3}\dd x < \infty,\qquad
\int_\R (f_Y(y)w_Y(y))^{1/3}\dd y < \infty.
\]
Moreover, the conditional mean functions $\mu_{Y|X}(x)\defeq \E[Y\mid X=x]$ and $\mu_{X|Y}(y)\defeq \E[X\mid Y=y]$ are continuously differentiable and the integrability conditions used in Appendix~\ref{app:proof_decoupling} and Appendix~\ref{app:proof_scalar_theorem} hold.\footnote{These conditions are mild for common smooth models (e.g., jointly Gaussian pairs) and are stated explicitly where they are used.}
\end{assumption}

Now denote the integrals in \cref{as:regularity} as
\begin{equation}\label{eq:IxIy_def}
I_X \defeq \int_{-\infty}^{\infty} (f_X(x)w_X(x))^{1/3}\dd x,\qquad
I_Y \defeq \int_{-\infty}^{\infty} (f_Y(y)w_Y(y))^{1/3}\dd y.
\end{equation}
An application of Cauchy-Schwartz gives the optimal high-rate quantizer for matrix multiplication.

\begin{theorem}[High-rate optimal matrix multiplication MSE]\label{thm:main_general}
Under Assumptions~\ref{as:pair_iid}--\ref{as:regularity},
\begin{equation}\label{eq:main_general_limit}
\inf_{Q_X,Q_Y}\ \mathcal{E}(Q_X,Q_Y)
=
mnk\left(\frac{I_X^3}{12\,K_X^2}+\frac{I_Y^3}{12\,K_Y^2}\right)
+ o\!\left(\frac{1}{K_X^2}+\frac{1}{K_Y^2}\right).
\end{equation}
Moreover, asymptotically optimal $K_X$- and $K_Y$-level quantizers are companding quantizers
with point densities
\begin{equation}\label{eq:lambda_general}
\lambda_X^\star(x)=\frac{(f_X(x)w_X(x))^{1/3}}{I_X},\qquad
\lambda_Y^\star(y)=\frac{(f_Y(y)w_Y(y))^{1/3}}{I_Y}.
\end{equation}
\end{theorem}

\begin{proof}
Appendix~\ref{app:proof_decoupling} proves the decoupling \eqref{eq:two_scalar_intro} with a remainder $o(K_X^{-2}+K_Y^{-2})$ and shows that the bias term in \eqref{eq:sum_variance_intro} is negligible at the $K^{-2}$ scale.
Appendix~\ref{app:proof_scalar_theorem} proves a weighted scalar high-rate theorem (including both achievability and converse) and identifies the unique optimizing point densities.
Combining these two appendices yields \eqref{eq:main_general_limit} and \eqref{eq:lambda_general}.
\end{proof}

\begin{corollary}[Rate form and optimal bit split]\label{cor:bit_split}
Let $R_X=\logtwo K_X$ and $R_Y=\logtwo K_Y$. Then the leading term in \eqref{eq:main_general_limit} can be written as
\[
\inf_{Q_X,Q_Y}\ \mathcal{E}(Q_X,Q_Y)
=
mnk\left(\alpha_X\,2^{-2R_X}+\alpha_Y\,2^{-2R_Y}\right)+o(2^{-2R_X}+2^{-2R_Y}),
\]
where $\alpha_X\defeq I_X^3/12$ and $\alpha_Y\defeq I_Y^3/12$.
If the \emph{pair rate} $R\defeq R_X+R_Y$ is fixed and large, then the minimizer of the leading term satisfies
\[
R_X^\star=\frac{R}{2}+\frac{1}{4}\logtwo\!\left(\frac{\alpha_X}{\alpha_Y}\right),\qquad
R_Y^\star=\frac{R}{2}-\frac{1}{4}\logtwo\!\left(\frac{\alpha_X}{\alpha_Y}\right),
\]
equivalently $K_X/K_Y=(\alpha_X/\alpha_Y)^{1/4}$.
\end{corollary}

\subsection{Special Case: Correlated Gaussian}\label{sec:gaussian}

Assume the multiplicative pair $(X,Y)$ is bivariate Gaussian:
\begin{equation}\label{eq:bivar_gauss}
\begin{bmatrix}X\\Y\end{bmatrix}
\sim \mathcal{N}\!\left(
\begin{bmatrix}0\\0\end{bmatrix},
\begin{bmatrix}
\sigma_X^2 & \rho\sigma_X\sigma_Y\\
\rho\sigma_X\sigma_Y & \sigma_Y^2
\end{bmatrix}\right),\qquad \rho\in(-1,1).
\end{equation}

Applying \cref{thm:main_general} to the Gaussian case enables a specialization to Gaussian joint densities. See Appendix~\ref{app:proof_closedform} for a proof.
\begin{corollary}
[Closed-form asymptotically optimal point density]\label{thm:closedform_density}
Under \eqref{eq:bivar_gauss}, the optimal companding point densities in \eqref{eq:lambda_general} take the closed form
\begin{align*}
\lambda_X^\star(x)
&=
\frac{
\exp\!\left(-\frac{x^2}{6\sigma_X^2}\right)\Bigl((1-\rho^2)+\rho^2\frac{x^2}{\sigma_X^2}\Bigr)^{1/3}
}{
\displaystyle \int_{-\infty}^{\infty}
\exp\!\left(-\frac{t^2}{6\sigma_X^2}\right)\Bigl((1-\rho^2)+\rho^2\frac{t^2}{\sigma_X^2}\Bigr)^{1/3}\dd t
},\\
\lambda_Y^\star(y)
&=
\frac{
\exp\!\left(-\frac{y^2}{6\sigma_Y^2}\right)\Bigl((1-\rho^2)+\rho^2\frac{y^2}{\sigma_Y^2}\Bigr)^{1/3}
}{
\displaystyle \int_{-\infty}^{\infty}
\exp\!\left(-\frac{t^2}{6\sigma_Y^2}\right)\Bigl((1-\rho^2)+\rho^2\frac{t^2}{\sigma_Y^2}\Bigr)^{1/3}\dd t
}.
\end{align*}
\end{corollary}

If we write the normalized variable as $u=x/\sigma_X$, we have
\[
\lambda^\star(u)\ \propto\ \exp\!\left(-\frac{u^2}{6}\right)\bigl((1-\rho^2)+\rho^2u^2\bigr)^{1/3}.
\]
The density has a qualitative difference compared to Gaussian companders: specifically, the optimal $\lambda^{*}$ goes through a phase transition as $\rho$ increases: at first it has a single mode, but as $\rho$ increases two modes appear. See \cref{fig:lambda} for a schematic, and Appendix~\ref{app:proof_phase_transition} for a proof. 

\begin{theorem}[Unimodal-to-bimodal transition at $\abs{\rho}=1/\sqrt{3}$]\label{thm:phase_transition}
Let $\lambda^\star(u)$ be the optimal density shape above.
\begin{enumerate}
\item If $\rho^2<1/3$, then $\lambda^\star(u)$ is unimodal with a unique global maximum at $u=0$.
\item If $\rho^2>1/3$, then $\lambda^\star(u)$ is bimodal: $u=0$ is a strict local minimum and the two symmetric global maxima occur at
\[
u_{\mathrm{peak}}=\pm\sqrt{3-\frac{1}{\rho^2}}.
\]
\item If $\rho^2=1/3$, the curvature at $u=0$ vanishes (critical splitting point).
\end{enumerate}
\end{theorem}

\begin{figure}[h]
    \centering
\includegraphics[width=0.8\linewidth]{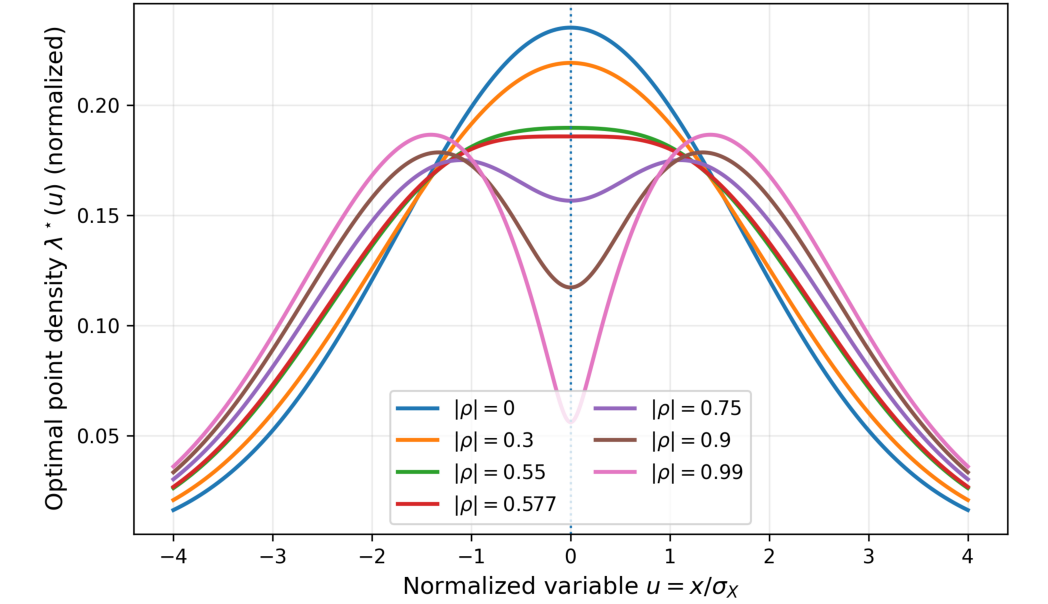}
    \caption{Optimal density phase transition. When $\rho=0$, there is only a single mode. As $\rho$ increases, an additional mode emerges at the critical value $\abs{\rho}=1/\sqrt{3}$.}
    \label{fig:lambda}
\end{figure}

At last, direct substitution gives the high-rate loss for the proposed matmul quantization.
\begin{corollary}[Gaussian high-rate constant for matrix multiplication MSE]\label{cor:gaussian_constant}
Let $K_X=K_Y=K$. Define
\[
J(\rho)\defeq \int_{-\infty}^{\infty}
\exp\!\left(-\frac{u^2}{6}\right)\bigl((1-\rho^2)+\rho^2u^2\bigr)^{1/3}\dd u.
\]
Then
\[
\lim_{K\to\infty}K^2\ \inf_{Q_X,Q_Y}\ \E\big[\norm{AB-\widehat A\widehat B}_\F^2\big]
=
\frac{mkn\,\sigma_X^2\sigma_Y^2}{6\sqrt{2\pi}}\ J(\rho)^3.
\]
A closed form for $J(\rho)$ is given in Appendix~\ref{app:normalizer}.
\end{corollary}

\section{Experiments}\label{sec:Exps}

We now empirically evaluate the performance of our derived quantization method with a variety of other state-of-the-art quantization methods on synthetic and real-world data.

\subsection{Synthetic Experiments: Matrix Multiplication}

First, we compare the performance of our optimal matrix multiplication quantizer with various commonly used quantizers for synthetically generated matrices that conform to our correlated Gaussian model. For this experiment, we test a variety of $\rho$ values and compare the quantization error in relative Frobenius norm for the matrices. 

\subsubsection{Alternate Quantizer Descriptions}

\begin{enumerate}
    \item \textbf{MatMul-Opt (Ours).}
    The theory-optimal companding quantizer derived from Corollary \ref{thm:closedform_density} for the correlated Gaussain case.
    Bin boundaries $\{t_i\}$ are placed by inverting the CDF of the optimal point density
    \begin{equation}
        \lambda^*(u) \propto
        \exp\!\left(-\tfrac{u^2}{6}\right)
        \left[(1-\rho^2) + \rho^2 u^2\right]^{1/3},
        \qquad u = x/\sigma,
    \end{equation}
    so that equal probability mass of $\lambda^*$ falls in each bin. 
    Separate quantizers are built for $A$ and $B$ using their empirical standard deviations $\hat\sigma_A$ and $\hat\sigma_B$ and the known correlation $\rho$.
    \item \textbf{Gaussian Compander.}
    A companding quantizer whose point density is the $\rho = 0$ special case of the optimal density from standard scalar high-rate theory.
    \item \textbf{Lloyd-Max.}
    The iteratively computed optimal scalar quantizer for a Gaussian source $X \sim \mathcal{N}(0,\sigma^2)$ \cite{lloyd1982least}. Starting from quantile-midpoint initialization, the algorithm alternates between (i) setting boundaries to midpoints of adjacent reconstruction levels and (ii) setting levels to the conditional mean of $|phi$ over each bin, until convergence in $\ell^\infty$ norm up to $200$ iterations.
    \item \textbf{Uniform.}
    A symmetric uniform quantizer with $K$ evenly-spaced reconstruction levles on $[-c,c]$. The clip value $c$ is selected by a grid search over $c \in [1.5\hat\sigma,\, 5.0 \hat\sigma]$ (36 points) to minimize empirical MSE on entries of the matrix being quantized. Separate clip values are calibrated for $A$ and $B$.
    \item \textbf{$\mu$-Law.}
    The ITU-T G.711 $\mu$-law compander \cite{jayant1984digital} with $\mu = 255$, adapted to continuous-valued inputs. 
    The companding function $f(x) = \operatorname{sgn}(x)\,
    \tfrac{\ln(1+\mu|x/c|)}{\ln(1+\mu)}$ maps inputs to a uniform grid of $K$
    levels in the logarithmically compressed domain; the inverse map recovers
    reconstruction levels in the original domain.
    The clip is fixed at $c = 4\hat\sigma$.
    $\mu$-law compression places more quantization levels near zero, which
    improves SNR for speech-like signals but is not matched to the
    matrix-multiplication task.
    \item \textbf{A-Law.}
    The ITU-T G.711 A-law compander~ \cite{jayant1984digital} with $A = 87.6$,
    also adapted to continuous inputs.
    The two-piece companding function is linear for $|x/c| < 1/A$ and
    logarithmic otherwise:
    \begin{align*}
    f(x) = \operatorname{sgn}(x) \cdot
    \begin{cases}
    \dfrac{A\,|x/c|}{1+\ln A} & |x/c| < \tfrac{1}{A},\\[6pt]
    \dfrac{1+\ln(A\,|x/c|)}{1+\ln A} & \tfrac{1}{A} \le |x/c| \le 1.
    \end{cases}
    \end{align*}
    $K$ uniform levels are placed in the companded domain with clip $c=4\hat\sigma$.
    A-law and $\mu$-law share the same logarithmic motivation and serve as
    classical telecommunications baselines.
    \item \textbf{NF4 (Normal Float 4-bit).}
    A $K$-level fixed codebook whose points are the quantile midpoints of
    $\mathcal{N}(0,1)$: $c_i = \Phi^{-1}\!\bigl((i+\tfrac{1}{2})/K\bigr)$
    for $i=0,\ldots,K-1$, normalized so that $\max_i|c_i|=1$, with endpoints
    clamped to $\pm 1$.
    This is an approximation of the QLoRA NF4 codebook~\cite{dettmers2023qlora},
    which places levels at Gaussian quantile midpoints to minimize MSE for
    normally-distributed weights.
    A per-matrix scalar scale is selected by grid search over
    $[0.5\hat\sigma,\,4.0\hat\sigma]$ (60 points) to minimize empirical MSE.
    \item \textbf{NV FP4 (E2M1).}
    NVIDIA's 4-bit floating-point format \cite{nvidia2025nvfp4} with
    1 sign bit, 2 exponent bits, and 1 mantissa bit with exponent bias 1.
    This yields 15 distinct finite values:
    $\{0,\,\pm 0.5,\,\pm 1.0,\,\pm 1.5,\,\pm 2.0,\,\pm 3.0,\,\pm 4.0,\,\pm 6.0\}$.
    Nearest-neighbor quantization is performed after scaling the codebook by a
    per-matrix scalar selected by grid search over
    $[0.25\hat\sigma,\,4.0\hat\sigma]$ (60 points) to minimize empirical MSE.
\end{enumerate}

\subsubsection{Results}

We obtained the following results averaged over $500$ generated $A \in \R^{128 \times 256}$ and $B \in \R^{256 \times 128}$ according to our correlated Gaussian model. From this, we can see that our optimal matrix multiplication quantizer consistently outperforms other commonly used quantization methods for our specified tasks.

\begin{figure}[h]
    \centering
    \includegraphics[width=0.75\linewidth]{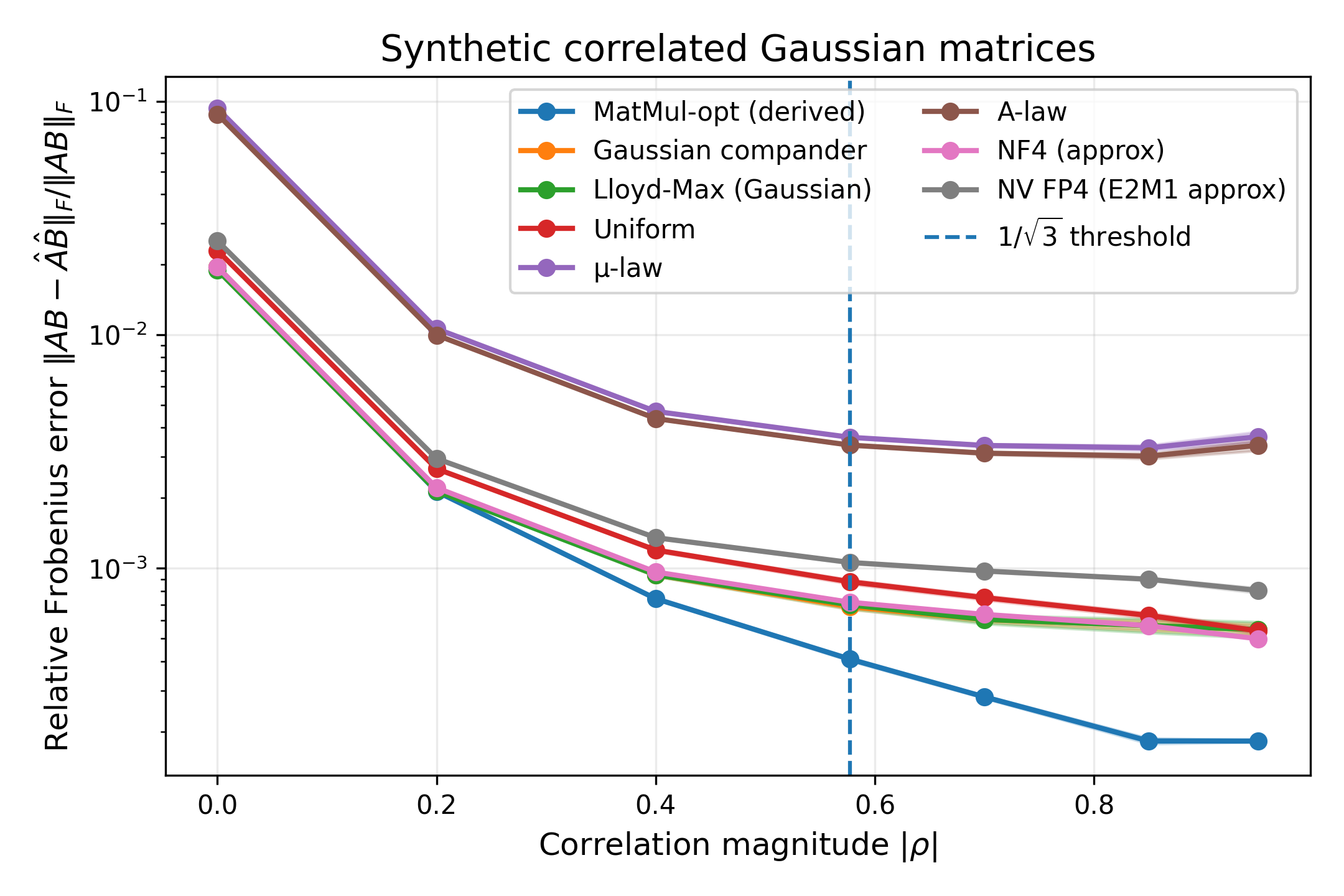}
    \caption{Performance of our optimal quantizer vs. other commonly used quantizers.}
    \label{fig:synth_results}
\end{figure}

We include $95\%$ error bars but used enough synthetic data that they are barely visible.

\subsection{Application in Quantized Least Squares}
We consider solving the least squares problem
\[
\min_{W} \|XW - Y\|_F^2,
\]
where $X \in \R^{n \times d}, W, W^{*} \in \R^{d \times m}$, $Y = XW^{*} + \epsilon Z$ and the entries of $X, W^{*}$ satisfy \cref{as:pair_iid}, and the joint distribution 
\begin{equation}\label{eq:bivar_gauss}
\begin{bmatrix}X\\W^{*}\end{bmatrix}
\sim \mathcal{N}\!\left(
\begin{bmatrix}0\\0\end{bmatrix},
\begin{bmatrix}
\sigma_X^2 & \rho\sigma_X\sigma_Y\\
\rho\sigma_X\sigma_Y & \sigma_Y^2
\end{bmatrix}\right),\qquad \rho = 0.6.
\end{equation}
Imagine $X$ is so large so that it cannot fit into a single GPU. A natural solution to this would be distributing the matrix over multiple GPUs then solving the problem by appropriately aggregating subproblems, but it may be timely as the loading/deloading time of the data matrix could be a bottleneck in such settings. Hence quantizing $X$ and $Y$ to load them onto a GPU then solving the quantized least squares,
\[
\min_{W} \|Q_XW - Q_Y\|_F^2,
\]
may make sense when solving the problem faster could be much more important than solving the problem exactly.

We apply the proposed quantization to the above problem and show that we can have better accuracy in terms of $\|W - W^{*}\|_F$ under the same bit budget. We compare three different schemes, where 
\begin{itemize}
    \item Scheme 1: quantize X, Y with Gaussian high-rate quantizer, simply considering the marginal distribution.
    \item Scheme 2: sweep through possible $\rho \in [-1, 1]$ and plot the loss with the best $\rho$.
    \item Scheme 3: estimate $\rho$ with
    \[
    \hat{\rho} = \frac{1}{ndm}\sum_{i=1}^{d}\sum_{j=1}^{m} (X\bar{W})_{ij},
    \]
    where $\bar{W}$ is a solution to the subproblem 
    \[
    \min_{W} \| \bar{X}W - \bar{Y}\|_F^2.
    \]
    Here $\bar{X}$, $\bar{Y}$ are row subsampled matrices of $X, Y$.
\end{itemize} 
\begin{figure}
    \centering
    \includegraphics[width=0.8\linewidth]{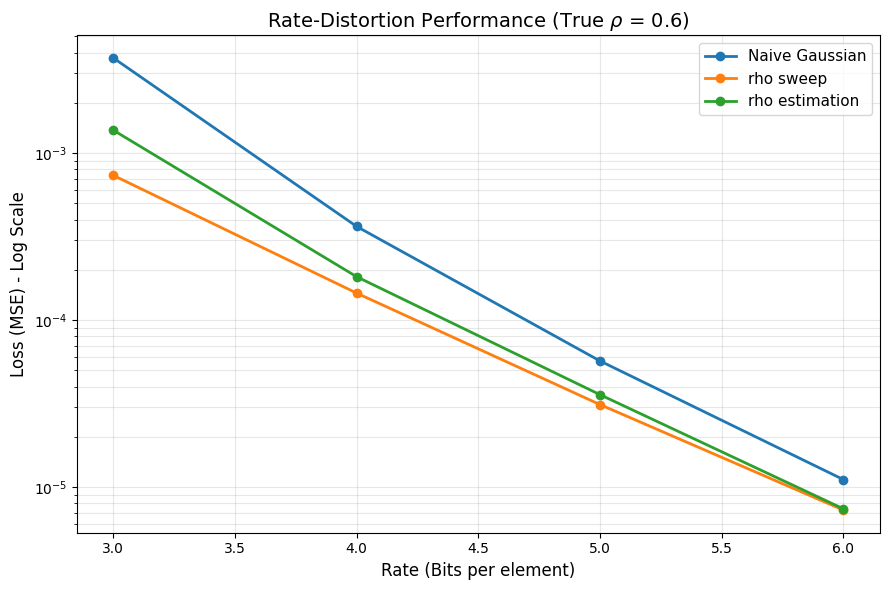}
    \caption{Comparison of three different schemes for solving quantized least squares. The $y$ axis is the difference between ground truth and the quantized solution $\|W - W^{*}\|_F$.}
    \label{fig:QLS}
\end{figure}
The results show that using nonzero $\rho$ can improve the accuracy of quantized least squares. Sweeping along the possible $\rho$s is better than using a $\rho$ estimate, especially for lower bit budgets. The main reason for this is that for bits=3,4, $\rho=-0.9$ is the optimal $\rho$ for quantized least squares, which is very different from the ground truth $\rho=0.6$. Such observation implies that there could be a different reason for why the correlation-aware quantization algorithm works better in the least squares setting. For higher rates the optimal $\rho$ becomes similar to ground truth $\rho$, and is better than naive Gaussian. 

\subsection{Quantization of Transformer-based Models}

Activation quantization is critical for efficient inference in Transformer-based models. We compare our method against INT8 quantization \cite{dettmers2022llmint88bitmatrixmultiplication} and FP8 quantization \cite{micikevicius2022fp8formatsdeeplearning}. In all cases, we apply per-token scaling to $Q$ and $K$ before quantization and use $K=256$ quantization levels (8 bits).

We empirically observe non-trivial correlation between entries of the query and key activations within attention heads (Figure~\ref{fig:per-head-rho}). This motivates applying our correlation-aware matrix multiplication quantizer independently to each attention head. For evaluation, we measure the relative Frobenius error of the pre-softmax attention logits:
\begin{align*}
    \frac{\left\lVert QK^\top - \hat{Q}\hat{K}^\top\right\rVert_F}{\left\lVert QK^\top \right\rVert_F},
\end{align*}
where $\hat Q$ and $\hat K$ denote quantized activations. This metric directly captures distortion in the attention logits.

\begin{figure}[h]
\centering
\includegraphics[width=0.75\linewidth]{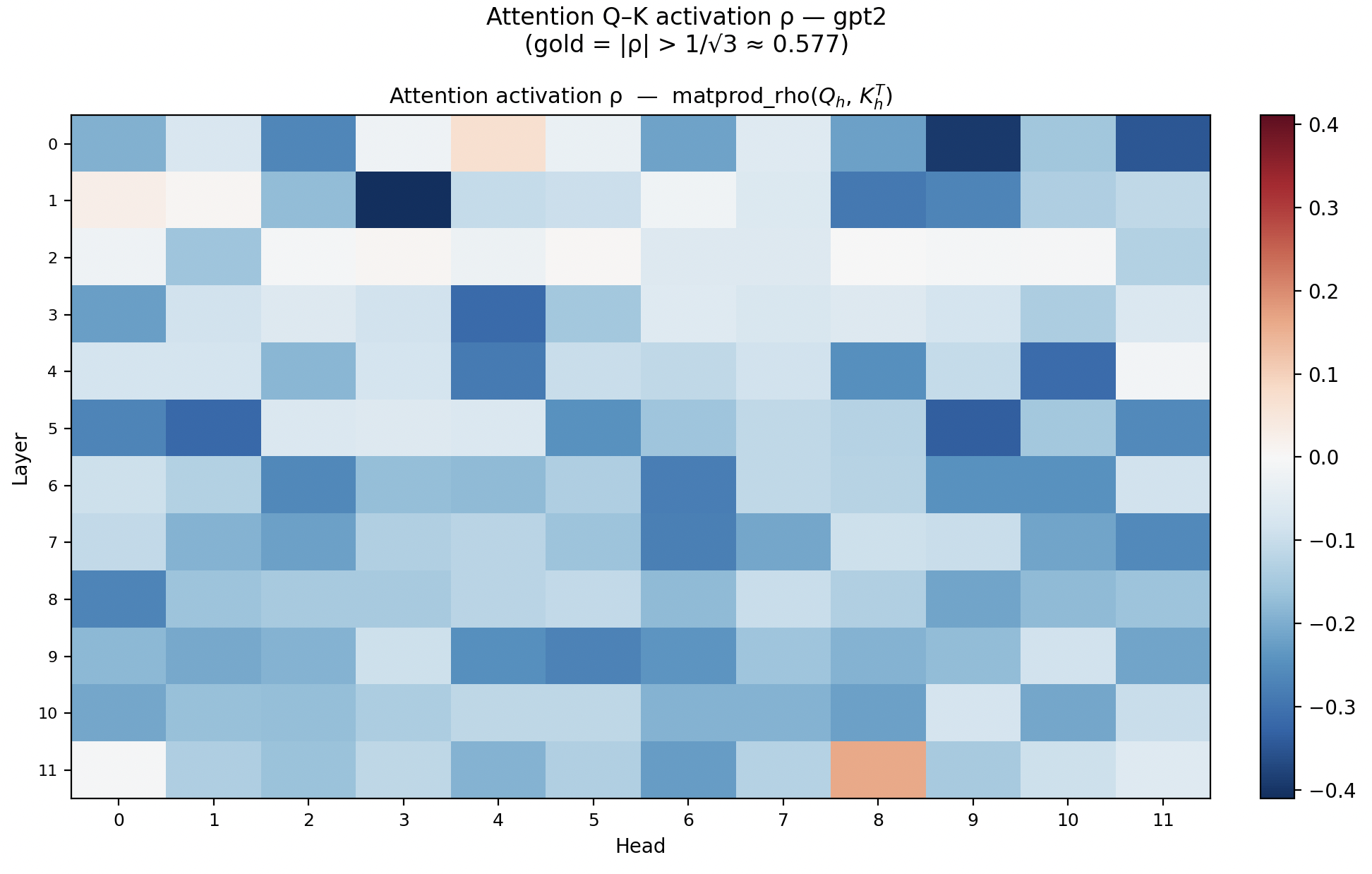}
\caption{Estimated $\rho$ value for each layer and head in GPT-2 Small}
\label{fig:per-head-rho}
\end{figure}

\subsubsection{Models and Data}
We evaluate on the GPT-2 family of models \cite{radford2019language} where there key and query activations are directly multiplied in addition so some Qwen3 models \cite{yang2025qwen3technicalreport} which apply rotary embeddings to the key and query activations before multiplying. For the Qwen3 models we tune $\rho$ based on the key and query activations after passing through the rotary embeddings.
Sequences are drawn from WikiText-2 \cite{merity2016pointersentinelmixturemodels}. We use 64 non-overlapping sequences of length 128 from the training split for evaluation. We hold out 32 sequences from the validation split for calibration. Per-head statistics are computed by concatenating all evaluation sequences.

\subsubsection{Activation Collection}
We extract $Q$ and $K$ at each attention layer using PyTorch forward hooks. For GPT-2, we hook the fused \texttt{c\_attn} module and split the output into $Q$, $K$, and $V$. For Qwen3, we hook the \texttt{q\_norm}, \texttt{k\_norm}, and \texttt{rotary\_emb} submodules per layer; the post-normalization $Q$ and $K$ are buffered, and rotary position embeddings are applied in NumPy to obtain post-RoPE activations.

\subsubsection{$\rho$-tuning}
For each attention head, we tune $\rho \in [0, 0.95]$ on the calibration set by grid search over 40 evenly spaced values, minimizing the relative Frobenius error above. For all methods, we use the same per-token $\ell_\infty$ scaling, i.e., we normalize each token vector so all entries lie in $[-1,1]$. Given a tuned $\rho$, we construct reproduction points using the closed-form density from Corollary~\ref{thm:closedform_density}.

\subsubsection{Results}
Table~\ref{tab:rho_tuned_wins} reports the fraction of heads for which $\rho$-tuned achieves lower relative Frobenius error than the baseline quantizer (a ``win'').

\begin{table}[t]
\centering
\caption{Win rate of $\rho$-tuned vs.\ INT8 and FP8 across GPT-2 sizes (win = lower relative Frobenius error per head).}
\label{tab:rho_tuned_wins}
\begin{tabular}{lcc}
\toprule
Model & $\rho$-tuned vs FP8 (\% wins) & $\rho$-tuned vs INT8 (\% wins) \\
\midrule
GPT-2 Small (12 layers, 12 heads)  & 100\% & 96.5\% \\
GPT-2 Medium (24 layers, 16 heads) & 100\% & 97.9\% \\
GPT-2 Large (36 layers, 20 heads)  & 100\% & 100\% \\
GPT-2 XL (48 layers, 25 heads) & 100\% & 99.7\%\\
Qwen3-0.6B (28 layers, 8 heads) & 98.7\% & 65.6\%\\
Qwen3-1.7B (28 layers, 8 heads)& 98.7\% & 59.4\%\\
Qwen3-8B (36 layers, 8 heads)& 86.1\% & 42.7\%\\
\bottomrule
\end{tabular}
\end{table}

Across model sizes, $\rho$-tuned matches or improves upon both INT8 and FP8, with the largest gains observed in larger models for the GPT-2 family of models while the opposite trend was observed on the Qwen3 family. We saw the performance degrade for larger Qwen3 models with the tuned $\rho$ algorithm losing to INT8 for Qwen3-8B. The tuned values of $\rho$ tended to be higher than the estimated values of $\rho$ (see figure \ref{fig:tuned-rho}), which indicates some level of model misspecification in our correlated Gaussian model. 

\begin{figure}[h]
    \centering
    \includegraphics[width=0.75\linewidth]{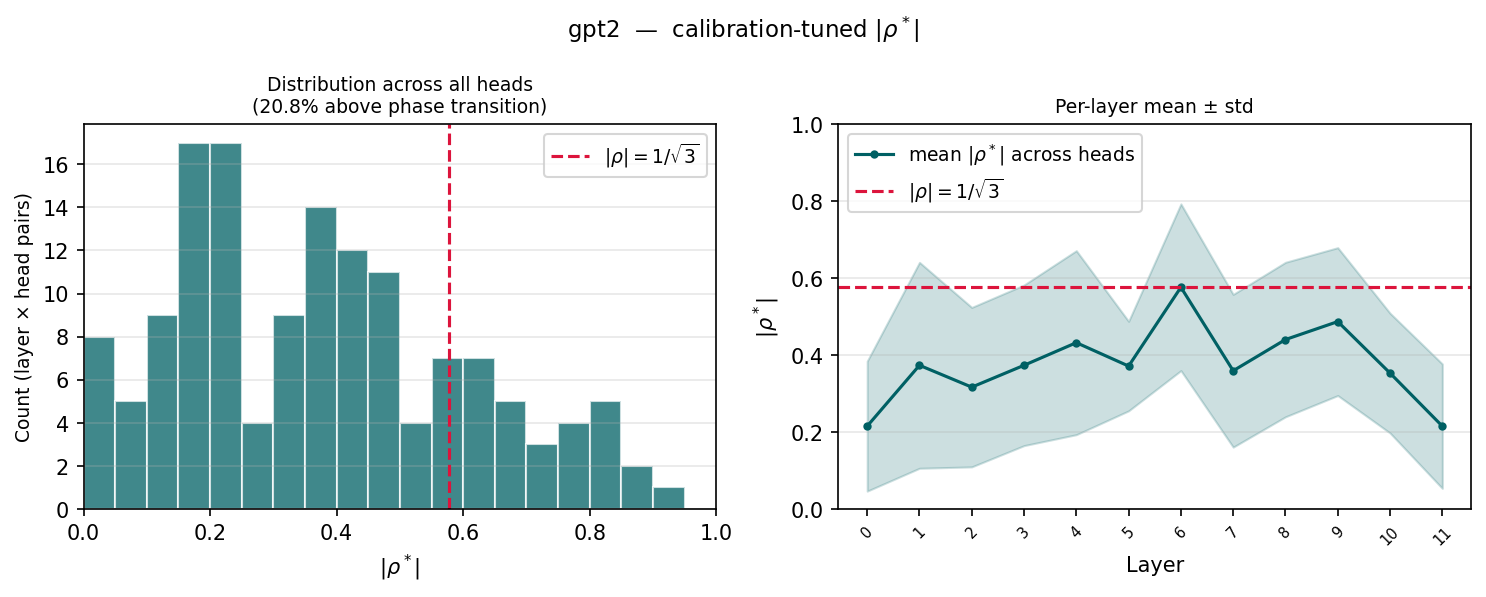}
    \caption{Tuned $|\rho|$ for GPT-small}
    \label{fig:tuned-rho}
\end{figure}

This misspecification seemed to hurt most in the Qwen3 family of models, which we posit is due to the nature of the rotary embeddings and their effect on the activation distributions. 

\paragraph{Limitations and Future Work.}
We plan to build upon our method so that it performs better on key and query quantization in models with rotary embeddings. We will also incorporate outlier-handling techniques commonly used in LLM quantization \cite{dettmers2022llmint88bitmatrixmultiplication} and report downstream task metrics in addition to logit-level distortion.
% \section{Discussion}
% The matrix-multiplication objective induces the conditional weights $w_X(x)=\E[Y^2\mid X=x]$ and $w_Y(y)=\E[X^2\mid Y=y]$, which can qualitatively alter the optimal point density compared to classical source MSE quantization. For correlated Gaussians, $w_X(0)=\sigma_Y^2(1-\rho^2)$ decreases as $\abs{\rho}\to 1$, making errors near $X\approx 0$ relatively unimportant for the product. Beyond $\abs{\rho}>1/\sqrt{3}$, this effect dominates the Gaussian marginal preference for allocating density near zero, yielding a bimodal ``dead-zone--like'' design.

% Information-theoretic lower bounds further clarify that if the sole goal is to estimate the product (or inner product), directly coding the target variable can offer a much faster $2^{-2R}$ distortion decay in total rate than factor quantization. The factor-quantize-and-multiply design is relevant precisely when factors must be stored/reused, or when computation requires quantized operands.

\section{Conclusion}
We derived a sharp high-rate characterization of optimal scalar quantization for matrix multiplication MSE under a pair-i.i.d.\ inner-product model. The leading error constant decomposes into two weighted scalar quantization problems with weights given by conditional second moments, leading to explicit optimal companding densities. For correlated Gaussian multiplicative pairs we obtained a closed-form density and proved a sharp unimodal-to-bimodal phase transition at $\abs{\rho}=1/\sqrt{3}$. We also benchmarked our derived quantizer on synthetically generated matrices for matrix multiplication and quantized least squares as well as on key and query activation quantization of the GPT-2 family of models.

% We also added information-theoretic benchmarks (entropy power / Shannon lower bounds) for product and inner-product reconstruction, highlighting the gap between coding the product and coding the factors.

\newpage
\appendix

\section{Proof of the Decoupling and High-Rate Expansion}\label{app:proof_decoupling}
This appendix provides a complete proof of the reduction \eqref{eq:two_scalar_intro} and the negligibility of the bias term in \eqref{eq:sum_variance_intro} at the $K^{-2}$ scale.

\subsection{Exact variance decomposition}
Let $D_\ell \defeq X_\ell Y_\ell - \widehat X_\ell \widehat Y_\ell$. Then
\[
S-\widehat S=\sum_{\ell=1}^k D_\ell.
\]
Because $\{(X_\ell,Y_\ell)\}$ are i.i.d.\ and quantized entrywise, the $D_\ell$ are i.i.d.\ with the same law as $D$. Hence
\[
\E[(S-\widehat S)^2]
=\sum_{\ell=1}^k \E[D_\ell^2] + 2\sum_{\ell<r} \E[D_\ell]\E[D_r]
= k\E[D^2] + k(k-1)(\E[D])^2,
\]
which is \eqref{eq:sum_variance_intro}.

\subsection{Quantizer model and elementary cell identities}\label{app:cell_identities}
Throughout this appendix we analyze companding quantizers $Q_X=Q_{K_X,\lambda_X}$ and $Q_Y=Q_{K_Y,\lambda_Y}$ constructed as in Section~\ref{sec:highrate_from_scratch}, with $\lambda_X,\lambda_Y$ continuously differentiable and strictly positive.

For $Q_{K,\lambda}$, let boundaries $x_i=G^{-1}(i/K)$ and reproduction points $r_i=G^{-1}((i-\tfrac12)/K)$, and denote the $i$th cell by $I_i=[x_{i-1},x_i)$ with width $\Delta_i=x_i-x_{i-1}$ and midpoint $m_i\defeq (x_{i-1}+x_i)/2$.

% Lemma by Calvin 2/27
\begin{lemma}[Cell width formula]\label{lem:cellwidth}
For each $i=1,\dots,K$, there exists $\xi_i\in I_i$ such that
\begin{equation}\label{eq:cellwidth_mvt}
\Delta_i=\frac{1}{K\,\lambda(\xi_i)}.
\end{equation}
\end{lemma}
\begin{proof}
By definition $G(x_i)-G(x_{i-1})=1/K$. Since $G$ is differentiable with derivative $G'(x)=\lambda(x)$, the mean value theorem gives a point $\xi_i\in(x_{i-1},x_i)$ such that
\[
\frac{1}{K}=G(x_i)-G(x_{i-1})=G'(\xi_i)(x_i-x_{i-1})=\lambda(\xi_i)\Delta_i.
\]
Rearranging yields \eqref{eq:cellwidth_mvt}.
\end{proof}

% Checked by Calvin 2/24
\begin{lemma}[Reproduction points are second-order close to cell midpoints]\label{lem:r_midpoint}
Assume $\lambda$ is continuously differentiable. Then for each fixed compact interval $[-M,M]$ there exists $C_M<\infty$ and $K_M$ such that for all $K\ge K_M$ and all cells $I_i$ intersecting $[-M,M]$,
\begin{equation}\label{eq:r_minus_m}
\abs{r_i-m_i}\le C_M\,\Delta_i^2.
\end{equation}
Consequently,
\begin{equation}\label{eq:first_second_moment_id}
\begin{aligned}
\int_{x_{i-1}}^{x_i} (x-r_i)\, \mathrm{d}x
&= \Delta_i(m_i-r_i)=O(\Delta_i^3), \\
\int_{x_{i-1}}^{x_i} (x-r_i)^2\, \mathrm{d}x
&= \frac{\Delta_i^3}{12}+O(\Delta_i^5).
\end{aligned}
\end{equation}
uniformly over such cells.
\end{lemma}
\begin{proof}
Let $g\defeq G^{-1}$ and $u\defeq (i-\tfrac12)/K$, so that $x_{i-1}=g(u-1/(2K))$, $x_i=g(u+1/(2K))$, and $r_i=g(u)$.
We know that $g$ is twice continuously differentiable by the inverse function theorem, so taking the second-order Taylor expansion of $g$ around $u$ gives
\[
g\!\left(u\pm \frac{1}{2K}\right)
= g(u)\pm \frac{g'(u)}{2K}+\frac{g''(\zeta_\pm)}{8K^2}
\]
for some $\zeta_\pm$ between $u$ and $u\pm 1/(2K)$.
Averaging the ``$+$'' and ``$-$'' expansions yields
\[
m_i=\frac{x_{i-1}+x_i}{2}=g(u)+\frac{g''(\zeta_+)+g''(\zeta_-)}{16K^2},
\]
hence
\[
m_i-r_i = \frac{g''(\zeta_+)+g''(\zeta_-)}{16K^2}.
\]

We note that by Lemma ~\ref{lem:cellwidth},
\[
\Delta_i = \frac{1}{K\lambda(\xi_i)}.
\]
Since $\lambda$ is continuous and strictly positive, it is bounded away from $0$ on $[-M,M]$, so $\Delta_i=\Theta(1/K)$ uniformly over cells intersecting $[-M,M]$.

Therefore, on a compact $[-M,M]$, the corresponding $u$ values lie in a compact subset of $(0,1)$ for all sufficiently large $K$ (since cells intersecting $[-M,M]$ have both endpoints in $g^{-1}([-M - C/K, M + C/K])$ as a consequence of our previous statement), so $g''$ is bounded there. Thus $\abs{m_i-r_i}=O(1/K^2) = O(\Delta_i^2)$ uniformly over cells intersecting $[-M,M]$, proving \eqref{eq:r_minus_m}.

Finally,

\begin{align*}
    \int_{x_{i-1}}^{x_i} (x-r_i)\dd x &=\left[\frac{(x-r_i)^2}{2}\right]_{x_{i-1}}^{x_i}
=\frac{(x_i-r_i)^2-(x_{i-1}-r_i)^2}{2}\\
&=\Delta_i(m_i-r_i)
\end{align*}

which together with \eqref{eq:r_minus_m} gives the first identity in \eqref{eq:first_second_moment_id}.
For the second identity, write $(x-r_i)^2=(x-m_i)^2+(m_i-r_i)^2+2(x-m_i)(m_i-r_i)$ and integrate over $[x_{i-1},x_i]$. The cross term integrates to zero by symmetry around $m_i$, giving

\begin{align*}
    \int_{x_{i-1}}^{x_i} (x-r_i)^2\dd x
&=\int_{x_{i-1}}^{x_i}(x-m_i)^2\dd x + \Delta_i(m_i-r_i)^2\\
&=\frac{\Delta_i^3}{12}+O(\Delta_i^5),
\end{align*}
since $\int_{x_{i-1}}^{x_i}(x-m_i)^2\dd x=\Delta_i^3/12$ exactly and $(m_i-r_i)^2=O(\Delta_i^4)$.
\end{proof}

\subsection{Bias term is negligible at the $K^{-2}$ scale}\label{app:bias_negligible}
We prove the key estimate $\E[D]=O(K_X^{-2}+K_Y^{-2})$.

%% Look at final limit
\begin{lemma}[A weighted first-moment bound for companders]\label{lem:first_moment}
Let $X$ have density $f$ and let $g:\R\to\R$ be continuously differentiable. Let $Q_{K,\lambda}$ be a companding quantizer with continuously differentiable point density $\lambda>0$.
Assume that $q(x)\defeq g(x)f(x)$ is twice continuously differentiable.
Then
\begin{equation}\label{eq:first_moment_order}
\E\big[g(X)\,(X-Q_{K,\lambda}(X))\big]=O\!\left(\frac{1}{K^2}\right),
\end{equation}
as $K\to\infty$.
\end{lemma}
\begin{proof}
Let the cells be $I_i=[x_{i-1},x_i)$ with reproduction $r_i$, width $\Delta_i$, and midpoint $m_i$ as in Section~\ref{app:cell_identities}. Write $e(x)\defeq x-Q_{K,\lambda}(x)$, so that $e(x)=x-r_i$ for $x\in I_i$.
Then
\[
\E[g(X)e(X)] = \sum_{i=1}^K \int_{I_i} q(x)(x-r_i)\dd x.
\]
Fix $M>0$ and split the sum into cells intersecting $[-M,M]$ (``interior cells'') and the remaining tail cells (``exterior cells'').

\paragraph*{Step 1: Interior cells.}
For an interior cell $I_i\subset[-M-\eta,M+\eta]$ (for a fixed small $\eta>0$ and $K$ large enough), expand $q$ around the midpoint $m_i$:
\[
q(x)=q(m_i)+q'(m_i)(x-m_i)+\frac{1}{2}q''(\zeta_{i,x})(x-m_i)^2,
\]
for some $\zeta_{i,x}\in I_i$. Multiply by $(x-r_i)=(x-m_i)+(m_i-r_i)$ and integrate over $I_i$:
\[
\int_{I_i} q(x)(x-r_i)\dd x = T_{i,0}+T_{i,1}+T_{i,2},
\]
where
\begin{align*}
T_{i,0}&\defeq q(m_i)\int_{I_i}(x-r_i)\dd x,\\
T_{i,1}&\defeq q'(m_i)\int_{I_i}(x-m_i)(x-r_i)\dd x,\\
T_{i,2}&\defeq \frac{1}{2}\int_{I_i}q''(\zeta_{i,x})(x-m_i)^2(x-r_i)\dd x.
\end{align*}
We bound each term using Lemma~\ref{lem:r_midpoint}.

First, Lemma~\ref{lem:r_midpoint} gives $\int_{I_i}(x-r_i)\dd x=O(\Delta_i^3)$. Since $q(m_i)$ is bounded on $[-M-\eta,M+\eta]$, we have $T_{i,0}=O(\Delta_i^3)$ uniformly over interior cells.

Second, note that $(x-m_i)(x-r_i)=(x-m_i)^2+(m_i-r_i)(x-m_i)$ and $\int_{I_i}(x-m_i)\dd x=0$ by symmetry. Hence
\[
\int_{I_i}(x-m_i)(x-r_i)\dd x=\int_{I_i}(x-m_i)^2\dd x=\frac{\Delta_i^3}{12}.
\]
Since $q'(m_i)$ is bounded on $[-M-\eta,M+\eta]$, it follows that $T_{i,1}=O(\Delta_i^3)$ uniformly over interior cells.

Third, on $I_i$ we have $\abs{x-m_i}\le \Delta_i/2$ and $\abs{x-r_i}\le \abs{x-m_i}+\abs{m_i-r_i}\le \Delta_i/2 + O(\Delta_i^2)=O(\Delta_i)$. Since $q''$ is bounded on $[-M-\eta,M+\eta]$, we obtain

\begin{align*}
\abs{T_{i,2}}
&\le C \int_{I_i} \abs{x-m_i}^2\,\abs{x-r_i}\dd x\\
&\le C' \Delta_i \int_{I_i} (x-m_i)^2 \dd x\\
&= C' \Delta_i\frac{\Delta_i^3}{12}\\
&=O(\Delta_i^4).
\end{align*}

Combining $T_{i,0}=O(\Delta_i^3)$, $T_{i,1}=O(\Delta_i^3)$, and $T_{i,2}=O(\Delta_i^4)$ gives
\[
\int_{I_i} q(x)(x-r_i)\dd x = O(\Delta_i^3)
\qquad\text{uniformly over interior cells.}
\]
Summing over all interior cells and using Lemma~\ref{lem:cellwidth} (which gives $\Delta_i=O(1/K)$ on compacts), we obtain

\begin{align*}
\sum_{i:\,I_i\cap[-M,M]\neq\emptyset}\int_{I_i} q(x)(x-r_i)\dd x
&= O\!\left(\sum_{i:\,I_i\cap[-M,M]\neq\emptyset}\Delta_i^3\right)\\
&= O\!\left(\frac{1}{K^2}\right),
\end{align*}

because there are $O(K)$ interior cells and each has $\Delta_i^3=O(1/K^3)$ uniformly.

\paragraph*{Step 2: Tail cells.}
Let $A_M=\{x:\abs{x}>M\}$. Using Cauchy--Schwarz,
\begin{align*}
\abs{\int_{A_M} q(x)(x-Q(x))\dd x}
&\le \left(\int_{A_M} q(x)^2 f(x)^{-1}\dd x\right)^{1/2} \\
&\quad \times \left(\int_{A_M} f(x)(x-Q(x))^2\dd x\right)^{1/2}.
\end{align*}
The second factor is the (unweighted) MSE on the tail set and is bounded by $\E[(X-Q(X))^2]^{1/2}=O(1/K)$ for companding quantizers (this is a special case of Appendix~\ref{app:proof_scalar_theorem} with $w\equiv 1$).
The first factor can be made arbitrarily small by choosing $M$ large enough because $q^2/f=g^2 f$ is integrable whenever $\E[g(X)^2]<\infty$, and in our application $g$ is a conditional mean with finite second moment by Assumption~\ref{as:pair_iid}.
Therefore, for any $\varepsilon>0$ we can pick $M$ so that the tail contribution is at most $\varepsilon/K$ in absolute value uniformly in $K$.

\paragraph*{Step 3: Combine.}
For this fixed $M$ and all sufficiently large $K$,
\[
\E[g(X)e(X)] = O\!\left(\frac{1}{K^2}\right)+\varepsilon\frac{1}{K}.
\]
Letting $\varepsilon\downarrow 0$ and then $K\to\infty$ yields \eqref{eq:first_moment_order}.
\end{proof}

%% Finish final section
\begin{proposition}[Bias order for the product error]\label{prop:bias_order}
Let $D=XY-\widehat X\widehat Y$ with $\widehat X=Q_{K_X,\lambda_X}(X)$ and $\widehat Y=Q_{K_Y,\lambda_Y}(Y)$ as above. Then
\begin{equation}\label{eq:ED_order}
\E[D]=O(K_X^{-2}+K_Y^{-2}).
\end{equation}
Consequently,
\[
k(k-1)\big(\E[D]\big)^2 = o(K_X^{-2}+K_Y^{-2})
\qquad\text{as}\qquad K_X,K_Y\to\infty.
\]
\end{proposition}
\begin{proof}
Write $e_X=X-\widehat X$ and $e_Y=Y-\widehat Y$. Using \eqref{eq:bilinear_error_intro},
\[
\E[D]=\E[Ye_X]+\E[Xe_Y]-\E[e_Xe_Y].
\]
Since $e_X$ is a function of $X$ only,
\[
\E[Ye_X]=\E\big[\E[Y\mid X]\,e_X\big]=\E[\mu_{Y|X}(X)\,e_X].
\]
Apply Lemma~\ref{lem:first_moment} with $g=\mu_{Y|X}$ and $Q_{K_X,\lambda_X}$ to obtain $\E[Ye_X]=O(K_X^{-2})$.
Similarly, $\E[Xe_Y]=O(K_Y^{-2})$.

For the final term, Cauchy--Schwarz gives
\[
\abs{\E[e_Xe_Y]}\le \sqrt{\E[e_X^2]\,\E[e_Y^2]}.
\]
The unweighted companding MSE satisfies $\E[e_X^2]=O(K_X^{-2})$ and $\E[e_Y^2]=O(K_Y^{-2})$ (a special case of Appendix~\ref{app:proof_scalar_theorem}), hence $\E[e_Xe_Y]=O((K_XK_Y)^{-1})$. Using $2/(K_XK_Y)\le 1/K_X^2+1/K_Y^2$ (AM--GM), this is also $O(K_X^{-2}+K_Y^{-2})$.
Combining the three bounds yields \eqref{eq:ED_order}.

Finally, squaring \eqref{eq:ED_order} gives $(\E[D])^2=O((K_X^{-2}+K_Y^{-2})^2)=o(K_X^{-2}+K_Y^{-2})$, completing the proof.
\end{proof}

\subsection{Dominant terms in $\E[D^2]$}\label{app:D2_dominant}
We now prove that the mixed terms in \eqref{eq:D2_expand_intro} have expectation $O(K_X^{-2}K_Y^{-2})$.

%% Analogous logic to Lemma 3 except in 2D, Checked by Calvin 2/25
\begin{lemma}[A generic mixed-cell bound]\label{lem:mixed_cell_bound}
Let $\phi:\R^2\to\R$ be twice continuously differentiable. Consider companding quantizers $Q_X=Q_{K_X,\lambda_X}$ and $Q_Y=Q_{K_Y,\lambda_Y}$ with continuosly differentiable point densities. Let $r_i$ and $s_j$ denote the reproduction points and let $\Delta_i$ and $\delta_j$ denote the corresponding cell widths for $Q_X$ and $Q_Y$ respectively.
Assume the function
\[
M_\phi(x,y)\defeq \max_{|\alpha|+|\beta|=2}\abs{\partial_x^\alpha\partial_y^\beta \phi(x,y)}
\]
is integrable with respect to $f_{X,Y}(x,y)\dd x\dd y$ when weighted by $(\lambda_X(x)\lambda_Y(y))^{-2}$ on compact sets (as made explicit in the proof below).
Then, as $K_X,K_Y\to\infty$,
\begin{equation}\label{eq:mixed_term_order_template}
\sum_{i=1}^{K_X}\sum_{j=1}^{K_Y}\int_{I_i}\int_{J_j}
\phi(x,y)\,(x-r_i)(y-s_j)\dd y\dd x
=O\!\left(\frac{1}{K_X^2K_Y^2}\right).
\end{equation}
\end{lemma}
\begin{proof}
Let $I_i=[x_{i-1},x_i)$ and $J_j=[y_{j-1},y_j)$ be the cells in $x$ and $y$. Fix $M>0$ and split the double sum into (i) rectangles intersecting $[-M,M]^2$ and (ii) rectangles in the complement. We bound both contributions.

\paragraph*{Step 1: Rectangles intersecting a compact set.}
For a rectangle $I_i\times J_j$ intersecting $[-M,M]^2$, pick the center point $(r_i,s_j)$ and perform a second-order Taylor expansion of $\phi$ around $(r_i,s_j)$:

\begin{align*}
    \phi(x,y) &=\phi(r_i,s_j)+\phi_x(r_i,s_j)(x-r_i)\\
    &+\phi_y(r_i,s_j)(y-s_j)+R_{i,j}(x,y),
\end{align*}
where the remainder satisfies
\[
\abs{R_{i,j}(x,y)}
\le \frac{1}{2} M_\phi(\zeta_{i,j}(x,y))\bigl((x-r_i)^2+(y-s_j)^2\bigr)
\]
for some $\zeta_{i,j}(x,y)\in I_i\times J_j$.

Multiply the Taylor expansion by $(x-r_i)(y-s_j)$ and integrate over $I_i\times J_j$. The constant term contributes
\[
\phi(r_i,s_j)\left(\int_{I_i}(x-r_i)\dd x\right)\left(\int_{J_j}(y-s_j)\dd y\right).
\]
By Lemma~\ref{lem:r_midpoint}, each one-dimensional integral is $O(\Delta_i^3)$ and $O(\delta_j^3)$ on a compact region, hence the constant term contribution is $O(\Delta_i^3\delta_j^3)$.

The $\phi_x$ term contributes
\[
\phi_x(r_i,s_j)\left(\int_{I_i}(x-r_i)^2\dd x\right)\left(\int_{J_j}(y-s_j)\dd y\right)
=O(\Delta_i^3\delta_j^3),
\]
since $\int_{I_i}(x-r_i)^2\dd x=O(\Delta_i^3)$ by \eqref{eq:first_second_moment_id}. The $\phi_y$ term is analogous.

For the remainder term, note that on $I_i\times J_j$ we have $\abs{x-r_i}=O(\Delta_i)$ and $\abs{y-s_j}=O(\delta_j)$, so
\[
\abs{R_{i,j}(x,y)(x-r_i)(y-s_j)}
\le C\, M_\phi(\zeta_{i,j}(x,y))\bigl(\Delta_i^3\delta_j+\Delta_i\delta_j^3\bigr).
\]
Integrating over $I_i\times J_j$ contributes at most

\begin{align*}
C\left(\Delta_i^3\delta_j+\Delta_i\delta_j^3\right)&\int_{I_i}\int_{J_j} M_\phi(\zeta_{i,j}(x,y))\dd y\dd x
\le\\
&C'\Delta_i\delta_j\left(\Delta_i^2+\delta_j^2\right)\sup_{I_i\times J_j}M_\phi.
\end{align*}

On compact sets, $\sup_{I_i\times J_j}M_\phi$ is bounded and $\Delta_i,\delta_j=O(1/K_X),O(1/K_Y)$, so this remainder contribution is also $O(\Delta_i^3\delta_j^3)$.

Combining all pieces, we have shown that for rectangles intersecting $[-M,M]^2$,
\[
\int_{I_i}\int_{J_j}\phi(x,y)(x-r_i)(y-s_j)\dd y\dd x = O(\Delta_i^3\delta_j^3),
\]
uniformly. Summing over the $O(K_XK_Y)$ rectangles intersecting $[-M,M]^2$ and using $\Delta_i=O(1/K_X)$, $\delta_j=O(1/K_Y)$ on compacts yields a total compact contribution of order $O(K_XK_Y\cdot K_X^{-3}K_Y^{-3})=O(K_X^{-2}K_Y^{-2})$.

\paragraph*{Step 2: Tail rectangles.}
On the complement of $[-M,M]^2$ we bound the entire integral by absolute values. Since $\abs{x-r_i}\le \abs{x}+\abs{r_i}$ and similarly for $y-s_j$, and since $(X,Y)$ has finite $(4+\epsilon)$ moments, by choosing $M$ large we can make the tail probability $\Pr((X,Y)\notin[-M,M]^2)$ arbitrarily small. The detailed bound follows the same truncation logic as in Lemma~\ref{lem:first_moment}: first control the tail by moments and then let $M\to\infty$. This yields a tail contribution that is $o(K_X^{-2}K_Y^{-2})$.

\paragraph*{Step 3: Combine and conclude.}

Let the rectangles which intersect $[-M,M]^2$ be $C := \{(i,j) : I_i \times J_j \cap [-M,M]^2 \neq \emptyset  \}$ and $\psi_{i,j}(x,y) := \phi(x,y)\,(x-r_i)(y-s_j)$. Then

\begin{align*}
    \sum_{i=1}^{K_X}\sum_{j=1}^{K_Y}\int_{I_i}\int_{J_j}\psi_{i,j}(x,y)
\dd y\dd x &= \sum_{(i,j) \in C} \int_{I_i}\int_{J_j}  
\psi_{i,j}(x,y)\dd y\dd x\\
&+  \sum_{(i,j) \not\in C} \int_{I_i}\int_{J_j} \psi_{i,j}(x,y) \dd y \dd x\\
\end{align*}

and from Step 1 we know that $\sum_{(i,j) \in C} \int_{I_i}\int_{J_j}  
\psi_{i,j}(x,y)\dd y\dd x = O(K_X^{-2} K_Y^{-2})$. 

From Step 2 we know that

\begin{align*}
    \int_{I_i}\int_{J_j} \psi_{i,j}(x,y) \dd y \dd x &= o(K_X^{-2} K_Y^{-2}), ~~ \forall (i,j) \not\in C\\
    \implies \sum_{(i,j) \not\in C} \int_{I_i}\int_{J_j} \psi_{i,j}(x,y) \dd y \dd x &= O(K_X K_Y) o(K_X^{-2} K_Y^{-2})\\
    &= O(K_X^{-2} K_Y^{-2}).
\end{align*}

Thus we see that

\begin{align*}
    \sum_{i=1}^{K_X}\sum_{j=1}^{K_Y}\int_{I_i}\int_{J_j}\psi_{i,j}(x,y)
\dd y\dd x &= O(K_X^{-2} K_Y^{-2})
\end{align*}

which is exactly \eqref{eq:mixed_term_order_template}.
\end{proof}

%% Similar to previous prop except 2D
\begin{proposition}[Mixed terms in $\E\lbrack D^2\rbrack$ are ${O(K_X^{-2}K_Y^{-2})}$]\label{prop:mixed_terms}
Let $D$ be as in Proposition~\ref{prop:bias_order}. Then the mixed terms in \eqref{eq:D2_expand_intro} satisfy
\begin{align}
\E[e_X^2e_Y^2] &= O(K_X^{-2}K_Y^{-2}),\label{eq:ex2ey2_order}\\
\E[XYe_Xe_Y] &= O(K_X^{-2}K_Y^{-2}),\label{eq:xyexey_order}\\
\E[Ye_X^2e_Y] &= O(K_X^{-2}K_Y^{-2}),\label{eq:yex2ey_order}\\
\E[Xe_Xe_Y^2] &= O(K_X^{-2}K_Y^{-2}).\label{eq:xexey2_order}
\end{align}
Consequently,

\begin{align*}
    \E[D^2]&=\E[Y^2e_X^2]+\E[X^2e_Y^2]+O(K_X^{-2}K_Y^{-2})\\
&=\E[Y^2e_X^2]+\E[X^2e_Y^2]+o(K_X^{-2}+K_Y^{-2}).
\end{align*}
\end{proposition}
\begin{proof}
Each expectation can be written as a double sum over rectangles $I_i\times J_j$.

\paragraph*{(i) $\E[XYe_Xe_Y]$.}
Write
\[
\E[XYe_Xe_Y]=\sum_{i=1}^{K_X}\sum_{j=1}^{K_Y}\int_{I_i}\int_{J_j} x y (x-r_i)(y-s_j) f_{X,Y}(x,y)\dd y\dd x.
\]
Apply Lemma~\ref{lem:mixed_cell_bound} with $\phi(x,y)=xy f_{X,Y}(x,y)$ to obtain \eqref{eq:xyexey_order}.

\paragraph*{(ii) $\E[Ye_X^2e_Y]$ and $\E[Xe_Xe_Y^2]$.}
For $\E[Ye_X^2e_Y]$ we write
\[
\E[Ye_X^2e_Y]=\sum_{i,j}\int_{I_i}\int_{J_j} y(x-r_i)^2(y-s_j) f_{X,Y}(x,y)\dd y\dd x.
\]
Fix $i$ and view $(x-r_i)^2$ as a bounded factor of order $O(\Delta_i^2)$ on $I_i$; then apply the same second-order Taylor argument in $y$ around $s_j$ (as in Lemma~\ref{lem:mixed_cell_bound}) to gain an extra factor $\delta_j^3$. Summing over $i,j$ yields $O(K_X^{-2}K_Y^{-2})$. The term $\E[Xe_Xe_Y^2]$ is symmetric.

\paragraph*{(iii) $\E[e_X^2e_Y^2]$.}
Similarly,
\[
\E[e_X^2e_Y^2]=\sum_{i,j}\int_{I_i}\int_{J_j} (x-r_i)^2(y-s_j)^2 f_{X,Y}(x,y)\dd y\dd x.
\]
On each rectangle, $(x-r_i)^2=O(\Delta_i^2)$ and $(y-s_j)^2=O(\delta_j^2)$, and integrating over the rectangle produces a factor $\Delta_i\delta_j$. Thus each rectangle contributes $O(\Delta_i^3\delta_j^3)$ and summing yields \eqref{eq:ex2ey2_order}.

\paragraph*{(iv) Conclusion.}
Insert \eqref{eq:ex2ey2_order}--\eqref{eq:xexey2_order} into \eqref{eq:D2_expand_intro}. Since $K_X^{-2}K_Y^{-2}=o(K_X^{-2}+K_Y^{-2})$, the stated expansion follows.
\end{proof}

\subsection{Conditional weights and conclusion}
Using Proposition~\ref{prop:mixed_terms} in \eqref{eq:D2_expand_intro} yields
\[
\E[D^2]=\E[Y^2e_X^2]+\E[X^2e_Y^2]+o(K_X^{-2}+K_Y^{-2}).
\]
Conditioning on $X$ and $Y$ gives
\[
\E[Y^2e_X^2]=\E[w_X(X)e_X^2],\qquad \E[X^2e_Y^2]=\E[w_Y(Y)e_Y^2],
\]
and combining with Proposition~\ref{prop:bias_order} in \eqref{eq:sum_variance_intro} yields \eqref{eq:two_scalar_intro}.

\section{Proof of the Weighted Scalar High-Rate Theorem}\label{app:proof_scalar_theorem}
We prove the weighted high-rate companding theorem used in Theorem~\ref{thm:main_general}, including a converse that rules out a better $K^{-2}$ constant for \emph{any} sequence of $K$-level scalar quantizers.

\subsection{Scalar setting}
Let $X$ have density $f$ on $\R$ and let $w:\R\to(0,\infty)$ be continuous. For a quantizer $Q$ with at most $K$ reproduction points, define the weighted MSE
\[
D(Q)\defeq \E\big[w(X)(X-Q(X))^2\big]=\int_{\R} f(x)w(x)(x-Q(x))^2\dd x.
\]
Let $D_K^\star\defeq \inf_{\abs{\mathrm{range}(Q)}\le K} D(Q)$ and define $h(x)\defeq f(x)w(x)$.

We assume throughout that $h$ is continuous and strictly positive, and that
\begin{equation}\label{eq:scalar_I_finite}
I\defeq \int_{\R} h(x)^{1/3}\dd x < \infty.
\end{equation}

%% Checked by Calvin 2/25
\subsection{Nearest-neighbor form}
\begin{lemma}[Nearest-neighbor regions]\label{lem:nn_app}
Fix reproduction points $r_1<\dots<r_K$. Among all quantizers using this codebook, the minimizer of $D(Q)$ assigns each $x$ to the nearest reproduction point in squared error. Hence optimal quantizers can be taken to have interval cells.
\end{lemma}
\begin{proof}

We note that by $w > 0$

\begin{align*}
w(x) (x - Q(x))^2 &\geq w(x) (x - r_{i_x})^2    
\end{align*}

where $i_x \defeq \arg\min_{i} \{(x - r_i)^2\}$ so that we achieve the minimum $D(Q)$ when $Q$ is chosen to assign $x$ to its nearest neighbor among the reproduction points. Furthermore, in one dimension, these intervals $ Q^{-1}(r_i) = [x_{i-1}, x_i)$ are characterized by $x_i \defeq \frac{r_i + r_{i+1}}{2}$.

\end{proof}

\subsection{Achievability: Bennett integral for a fixed point density}\label{app:achievability_bennett}
Fix a continuous point density $\lambda>0$ with $\int_{\R}\lambda=1$ and construct $Q_{K,\lambda}$. Denote its cells by $I_i=[x_{i-1},x_i)$ with width $\Delta_i=x_i-x_{i-1}$, midpoint $m_i$, and reproduction point $r_i$.

\begin{lemma}[Bennett integral limit for companders]\label{lem:bennett_limit}
Assume $\lambda$ is continuously differentiable and $\int_\R h(x)/\lambda(x)^2\dd x<\infty$. Then
\begin{equation}\label{eq:bennett_app_new}
\lim_{K\to\infty} K^2 D(Q_{K,\lambda})
=\frac{1}{12}\int_{\R}\frac{h(x)}{\lambda(x)^2}\dd x.
\end{equation}
\end{lemma}
\begin{proof}
Write
\[
D(Q_{K,\lambda})=\sum_{i=1}^K \int_{I_i} h(x)(x-r_i)^2\dd x.
\]
Fix $M>0$ and decompose the sum into cells that intersect $[-M,M]$ and cells that do not. We handle these two contributions separately.

\paragraph*{Step 1: Compact contribution.}
For a cell $I_i$ intersecting $[-M,M]$, Lemma~\ref{lem:r_midpoint} gives
\[
\int_{I_i} (x-r_i)^2\dd x = \frac{\Delta_i^3}{12}+O(\Delta_i^5),
\]
uniformly. Since $h$ is uniformly continuous on a slightly enlarged compact interval and $\Delta_i\to 0$ uniformly on that compact region (by Lemma~\ref{lem:cellwidth}), there exists $\xi_i\in I_i$ such that
\[
\int_{I_i} h(x)(x-r_i)^2\dd x = h(\xi_i)\left(\frac{\Delta_i^3}{12}+O(\Delta_i^5)\right).
\]
Multiply by $K^2$ and use Lemma~\ref{lem:cellwidth} to write $\Delta_i = 1/(K\lambda(\eta_i))$ for some $\eta_i\in I_i$:
\[
K^2\int_{I_i} h(x)(x-r_i)^2\dd x
= \frac{1}{12}\frac{h(\xi_i)}{\lambda(\eta_i)^2}\Delta_i + O\!\left(\frac{h(\xi_i)}{K^2}\Delta_i\right),
\]
because $K^2\Delta_i^5 = (K^2\Delta_i^3)\Delta_i^2 = O(\Delta_i^2)\Delta_i = O(K^{-2})\Delta_i$ on compacts where $\Delta_i=O(1/K)$.

Summing over all cells intersecting $[-M,M]$ gives
\begin{equation}\label{eq:riemann_sum_compact}
K^2\sum_{i:\,I_i\cap[-M,M]\neq\emptyset}\int_{I_i} h(x)(x-r_i)^2\dd x
= \frac{1}{12}\sum_{i:\,I_i\cap[-M,M]\neq\emptyset}\frac{h(\xi_i)}{\lambda(\eta_i)^2}\Delta_i + o_M(1),
\end{equation}
where $o_M(1)\to 0$ as $K\to\infty$ for each fixed $M$.

Because $\xi_i,\eta_i\in I_i$ and $\max_{i:\,I_i\cap[-M,M]\neq\emptyset}\Delta_i\to 0$, continuity implies $\frac{h(\xi_i)}{\lambda(\eta_i)^2} = \frac{h(\zeta_i)}{\lambda(\zeta_i)^2}+o(1)$ for some $\zeta_i\in I_i$. Therefore the right-hand side of \eqref{eq:riemann_sum_compact} is a Riemann sum for $\int_{-M}^{M} h(x)/\lambda(x)^2\dd x$, and we conclude that
\[
\lim_{K\to\infty}
K^2\sum_{i:\,I_i\cap[-M,M]\neq\emptyset}\int_{I_i} h(x)(x-r_i)^2\dd x
= \frac{1}{12}\int_{-M}^{M}\frac{h(x)}{\lambda(x)^2}\dd x.
\]

\paragraph*{Step 2: Tail contribution.}
Since $h/\lambda^2$ is integrable by assumption, we can choose $M$ large enough such that
\[
\int_{\abs{x}>M}\frac{h(x)}{\lambda(x)^2}\dd x \le \varepsilon.
\]
A similar bound (using nonnegativity of the integrand and the same cell-width identity as above) shows that for all $K$,
\[
0\le
K^2\sum_{i:\,I_i\cap[-M,M]=\emptyset}\int_{I_i} h(x)(x-r_i)^2\dd x
\le \frac{1}{12}\int_{\abs{x}>M}\frac{h(x)}{\lambda(x)^2}\dd x + o(1)
\le \frac{\varepsilon}{12}+o(1).
\]
Letting $\varepsilon\downarrow 0$ and then $K\to\infty$  gives that the tail contribution vanishes in the $K^2$-scaled limit.

\paragraph*{Step 3: Combine and conclude.}
Combining the compact and tail contributions and then letting $M\to\infty$ yields \eqref{eq:bennett_app_new}.
\end{proof}

\subsection{Optimize over $\lambda$ (H\"older inequality)}
Given Lemma~\ref{lem:bennett_limit}, the leading constant for a fixed point density $\lambda$ is
\[
J(\lambda)\defeq \int_{\R}\frac{h(x)}{\lambda(x)^2}\dd x,
\qquad \text{subject to}\qquad \lambda>0,\ \int_\R \lambda=1.
\]
Let $a(x)\defeq h(x)^{1/3}$. Write
\[
a(x)=\Bigl(\frac{a(x)}{\lambda(x)^{2/3}}\Bigr)\lambda(x)^{2/3}.
\]
H\"older with exponents $3$ and $3/2$ gives
\[
\int_{\R} a
\le \left(\int_{\R} \frac{a^3}{\lambda^2}\right)^{1/3}\left(\int_{\R}\lambda\right)^{2/3}
= J(\lambda)^{1/3}.
\]
Therefore $J(\lambda)\ge (\int a)^3 = I^3$, with equality if and only if $\lambda(x)\propto a(x)$, i.e.,
\begin{equation}\label{eq:lambda_star_scalar}
\lambda^\star(x)=\frac{h(x)^{1/3}}{\int_\R h(t)^{1/3}\dd t}=\frac{h(x)^{1/3}}{I}.
\end{equation}
Combining with Lemma~\ref{lem:bennett_limit} gives the \emph{achievability} statement
\begin{equation}\label{eq:achievability_scalar}
\limsup_{K\to\infty} K^2 D_K^\star \le \frac{I^3}{12}.
\end{equation}

\subsection{Converse: no smaller $K^{-2}$ constant is possible}\label{app:scalar_converse}
We now prove the matching lower bound $\liminf_{K\to\infty}K^2D_K^\star\ge I^3/12$.

%% Checked by Calvin 2/26
\begin{lemma}[A lower bound on the distortion over a compact interval]\label{lem:compact_lower_bound}
Fix $M>0$ and define the truncated distortion
\[
D_M(Q)\defeq \int_{-M}^{M} h(x)(x-Q(x))^2\dd x.
\]
Let $h_{\min,M}\defeq \inf_{x\in[-M,M]} h(x)$, which is strictly positive by continuity.
Then for any $K$-level quantizer $Q$,
\begin{equation}\label{eq:compact_lower_bound}
D_M(Q)\ \ge\ \frac{1}{12}\sum_{\ell=1}^{L} h_{\ell,\min}\,\abs{J_\ell}^3,
\end{equation}
where $\{J_\ell\}_{\ell=1}^{L}$ is the partition of $[-M,M]$ induced by the quantizer cells (each $J_\ell$ is an interval on which $Q$ is constant), $L\le K$, and $h_{\ell,\min}\defeq \inf_{x\in J_\ell} h(x)$.
\end{lemma}
\begin{proof}
On each induced interval $J_\ell$ the quantizer output is a constant reproduction point, say $Q(x)=r_\ell$. Since $h(x)\ge h_{\ell,\min}$ for $x\in J_\ell$,
\[
\int_{J_\ell} h(x)(x-r_\ell)^2\dd x
\ge h_{\ell,\min}\int_{J_\ell}(x-r_\ell)^2\dd x.
\]
For an interval of length $\abs{J_\ell}$, the function $r\mapsto \int_{J_\ell}(x-r)^2\dd x$ is minimized at the midpoint of $J_\ell$, with minimum value $\abs{J_\ell}^3/12$. Therefore the integral is at least $\abs{J_\ell}^3/12$ for any $r_\ell$, giving
\[
\int_{J_\ell} h(x)(x-r_\ell)^2\dd x
\ge \frac{h_{\ell,\min}}{12}\abs{J_\ell}^3.
\]
Summing over $\ell=1,\dots,L$ yields \eqref{eq:compact_lower_bound}.
\end{proof}

\begin{lemma}[Approximating $\int_{-M}^{M}h^{1/3}$ by a partition of small mesh]\label{lem:min_integral_approx}
Fix $M>0$ and let $\omega_M(\delta)$ denote the modulus of continuity of $h^{1/3}$ on $[-M,M]$:
\[
\omega_M(\delta)\defeq \sup_{\substack{x,y\in[-M,M]\\|x-y|\le \delta}}\abs{h(x)^{1/3}-h(y)^{1/3}}.
\]
Then for any partition $\{J_\ell\}_{\ell=1}^{L}$ of $[-M,M]$ into intervals with maximum length at most $\delta$,
\begin{equation}\label{eq:min_integral_approx}
\sum_{\ell=1}^{L} h_{\ell,\min}^{1/3}\,\abs{J_\ell}
\ \ge\ \int_{-M}^{M} h(x)^{1/3}\dd x\ -\ 2M\,\omega_M(\delta).
\end{equation}
\end{lemma}
\begin{proof}
Fix an interval $J_\ell$. For any $x\in J_\ell$, we have $h_{\ell,\min}^{1/3}\ge h(x)^{1/3}-\omega_M(\delta)$ because $\sup_{t\in J_\ell}\abs{h(t)^{1/3}-h(x)^{1/3}}\le \omega_M(\delta)$ and $\abs{J_\ell}\le\delta$.
Integrating over $x\in J_\ell$ yields
\[
h_{\ell,\min}^{1/3}\abs{J_\ell}
\ge \int_{J_\ell} h(x)^{1/3}\dd x - \omega_M(\delta)\abs{J_\ell}.
\]
Summing over $\ell$ and using $\sum_\ell \abs{J_\ell}=2M$ gives \eqref{eq:min_integral_approx}.
\end{proof}

\begin{lemma}[Compact-set converse constmmant]\label{lem:compact_converse_constant}
Fix $M>0$ and define
\[
I_M\defeq \int_{-M}^{M} h(x)^{1/3}\dd x.
\]
Then
\begin{equation}\label{eq:compact_converse_constant}
\liminf_{K\to\infty} K^2 D_K^\star \ \ge\ \frac{I_M^3}{12}.
\end{equation}
\end{lemma}
\begin{proof}
Since $D(Q)\ge D_M(Q)$ for every quantizer $Q$, we have $D_K^\star\ge \inf_{\abs{\mathrm{range}(Q)}\le K} D_M(Q)$, and it is enough to lower bound the latter.

Fix $\delta>0$. Consider any $K$-level quantizer $Q$, and let $\{J_\ell\}_{\ell=1}^{L}$ be the induced partition of $[-M,M]$ (as in Lemma~\ref{lem:compact_lower_bound}), with $L\le K$.

\paragraph*{Case 1: the partition mesh exceeds $\delta$.}
If $\max_\ell \abs{J_\ell}>\delta$, then Lemma~\ref{lem:compact_lower_bound} and $h_{\ell,\min}\ge h_{\min,M}$ imply
\[
D_M(Q)\ge \frac{h_{\min,M}}{12}\delta^3.
\]

\paragraph*{Case 2: the partition mesh is at most $\delta$.}
If $\max_\ell \abs{J_\ell}\le\delta$, then Lemma~\ref{lem:compact_lower_bound} gives
\[
D_M(Q)\ge \frac{1}{12}\sum_{\ell=1}^{L}\bigl(h_{\ell,\min}^{1/3}\abs{J_\ell}\bigr)^3.
\]
Apply H\"older to the nonnegative numbers $a_\ell\defeq h_{\ell,\min}^{1/3}\abs{J_\ell}$:
\[
\left(\sum_{\ell=1}^{L} a_\ell\right)^3 \le \left(\sum_{\ell=1}^{L} a_\ell^3\right) L^2 \le \left(\sum_{\ell=1}^{L} a_\ell^3\right) K^2,
\]
so $\sum_{\ell=1}^{L} a_\ell^3\ge (\sum_{\ell=1}^{L} a_\ell)^3/K^2$. Therefore
\[
D_M(Q)\ge \frac{1}{12K^2}\left(\sum_{\ell=1}^{L} h_{\ell,\min}^{1/3}\abs{J_\ell}\right)^3.
\]
By Lemma~\ref{lem:min_integral_approx}, the bracketed term is at least $I_M-2M\omega_M(\delta)$. Hence
\[
D_M(Q)\ge \frac{1}{12K^2}\bigl(I_M-2M\omega_M(\delta)\bigr)^3.
\]

\paragraph*{Combine the two cases.}
We have shown that for every $K$-level quantizer $Q$,
\[
D_M(Q)\ge \min\left\{\frac{h_{\min,M}}{12}\delta^3,\ \frac{1}{12K^2}\bigl(I_M-2M\omega_M(\delta)\bigr)^3\right\}.
\]
Taking the infimum over all $Q$ with at most $K$ levels preserves the inequality. Multiply by $K^2$ and let $K\to\infty$. For fixed $\delta$, the second term in the minimum dominates for large $K$, yielding
\[
\liminf_{K\to\infty} K^2 D_K^\star \ge \frac{1}{12}\bigl(I_M-2M\omega_M(\delta)\bigr)^3.
\]
Finally let $\delta\downarrow 0$. Since $h^{1/3}$ is uniformly continuous on $[-M,M]$, $\omega_M(\delta)\to 0$, and we obtain \eqref{eq:compact_converse_constant}.
\end{proof}

\begin{theorem}[Weighted scalar high-rate constant and optimal density]\label{thm:scalar_highrate}
Under \eqref{eq:scalar_I_finite},
\begin{equation}\label{eq:scalar_limit_final}
\lim_{K\to\infty} K^2 D_K^\star = \frac{1}{12}\left(\int_{\R} (f(x)w(x))^{1/3}\dd x\right)^3.
\end{equation}
Moreover, the unique minimizing point density for the Bennett functional is \eqref{eq:lambda_star_scalar}, and the sequence of companders $Q_{K,\lambda^\star}$ achieves the limit in \eqref{eq:scalar_limit_final}.
\end{theorem}
\begin{proof}
The achievability bound \eqref{eq:achievability_scalar} follows from Lemma~\ref{lem:bennett_limit} and the H\"older minimization.
For the converse, Lemma~\ref{lem:compact_converse_constant} implies that for every $M$,
\[
\liminf_{K\to\infty} K^2 D_K^\star \ge \frac{I_M^3}{12},
\qquad I_M=\int_{-M}^{M} h(x)^{1/3}\dd x.
\]
Letting $M\to\infty$ and using monotone convergence (since $h^{1/3}\ge 0$ and integrable) gives $\lim_{M\to\infty} I_M = I$, hence $\liminf_{K\to\infty}K^2 D_K^\star \ge I^3/12$. Together with \eqref{eq:achievability_scalar} this yields \eqref{eq:scalar_limit_final}. The optimality and uniqueness of $\lambda^\star$ follow from the equality condition in H\"older.
\end{proof}

\section{Proof of Theorem~\ref{thm:closedform_density}}\label{app:proof_closedform}
By Theorem~\ref{thm:main_general}, $\lambda_X^\star(x)\propto (f_X(x)w_X(x))^{1/3}$ and $\lambda_Y^\star(y)\propto (f_Y(y)w_Y(y))^{1/3}$.

For joint Gaussian \eqref{eq:bivar_gauss}, the conditional law satisfies
\[
Y\mid X=x \sim \mathcal{N}\!\left(\rho\frac{\sigma_Y}{\sigma_X}x,\ \sigma_Y^2(1-\rho^2)\right).
\]
Hence
\begin{align}
w_X(x)=\E[Y^2\mid X=x]
&=\Var(Y\mid X=x) + (\E[Y\mid X=x])^2 \nonumber\\
&=\sigma_Y^2(1-\rho^2)+\rho^2\frac{\sigma_Y^2}{\sigma_X^2}x^2. \label{eq:gauss_wx_app}
\end{align}
Similarly,
\begin{equation}\label{eq:gauss_wy_app}
w_Y(y)=\sigma_X^2(1-\rho^2)+\rho^2\frac{\sigma_X^2}{\sigma_Y^2}y^2.
\end{equation}
The marginal $X\sim\mathcal{N}(0,\sigma_X^2)$ has density $f_X(x)\propto \exp(-x^2/(2\sigma_X^2))$, so
\[
(f_X(x)w_X(x))^{1/3}\ \propto\ \exp\!\left(-\frac{x^2}{6\sigma_X^2}\right)\Bigl((1-\rho^2)+\rho^2\frac{x^2}{\sigma_X^2}\Bigr)^{1/3},
\]
and normalization yields the stated $\lambda_X^\star$. The same reasoning with \eqref{eq:gauss_wy_app} yields $\lambda_Y^\star$. \qedhere

\section{Proof of Theorem~\ref{thm:phase_transition}}\label{app:proof_phase_transition}
Let $u=x/\sigma_X$. Up to an additive constant, the log-density is
\[
\ell(u)= -\frac{u^2}{6}+\frac{1}{3}\log\big((1-\rho^2)+\rho^2u^2\big).
\]
Differentiate:
\[
\ell'(u)
= u\left[-\frac{1}{3}+\frac{2\rho^2}{3((1-\rho^2)+\rho^2u^2)}\right].
\]
Thus $u=0$ is always stationary. For $u\neq 0$, stationarity requires
\[
-\frac{1}{3}+\frac{2\rho^2}{3((1-\rho^2)+\rho^2u^2)}=0
\iff (1-\rho^2)+\rho^2u^2=2\rho^2
\iff u^2=3-\frac{1}{\rho^2}.
\]
Real nonzero stationary points exist iff $\rho^2>1/3$, and then they are $\pm\sqrt{3-1/\rho^2}$.

The curvature at the origin is
\[
\ell''(0)=\frac{3\rho^2-1}{3(1-\rho^2)}.
\]
If $\rho^2<1/3$, then $\ell''(0)<0$ and, with no other stationary points, the density is unimodal with maximum at $0$.
If $\rho^2>1/3$, then $\ell''(0)>0$ so $0$ is a strict local minimum and the only other stationary points are the two symmetric maxima computed above, yielding bimodality.
At $\rho^2=1/3$, $\ell''(0)=0$ (critical point). \qedhere

\section{Closed Form for the Normalizer Integral}\label{app:normalizer}
Define
\[
J(\rho)=\int_{-\infty}^{\infty} e^{-u^2/6}\big((1-\rho^2)+\rho^2u^2\big)^{1/3}\dd u.
\]
By even symmetry and the substitution $t=u^2$ (so $\dd u=\tfrac{1}{2}t^{-1/2}\dd t$),
\begin{equation}\label{eq:J_t}
J(\rho)=\int_0^\infty e^{-t/6}\big((1-\rho^2)+\rho^2 t\big)^{1/3} t^{-1/2}\dd t.
\end{equation}
Factoring $(1-\rho^2)$ and rescaling $t$ yields, for $\rho\neq 0$,
\begin{equation}\label{eq:J_U_app}
J(\rho)=
\frac{\sqrt{\pi}\,(1-\rho^2)^{5/6}}{\abs{\rho}}\, U\!\left(\frac12,\frac{11}{6},\frac{1-\rho^2}{6\rho^2}\right),
\end{equation}
where $U(a,b,z)$ is Tricomi's confluent hypergeometric function. Also $J(0)=\sqrt{6\pi}$.

% \section{Product PDF and Entropy Power for Correlated Gaussians}\label{app:prod_pdf_entropy}
% \subsection{Proof of Proposition \ref{prop:prod_pdf}}
% The density \eqref{eq:prod_pdf} is a known formula for the product of two zero-mean correlated normal random variables; see \cite{gaunt2022basic}. We record the result:
% \[
% f_Z(z)=\frac{1}{\pi s\sqrt{1-\rho^2}}\exp\!\left(\frac{\rho z}{s(1-\rho^2)}\right)
% K_0\!\left(\frac{|z|}{s(1-\rho^2)}\right).
% \]
% Scaling gives $h(aU)=h(U)+\log|a|$ in one dimension, hence $h(Z)=h(Z/s)+\log s$ and $N(Z)=s^2 N(Z/s)$ by \eqref{eq:entropy_power_def}.
% \subsection{Computable integral form for $h(Z)$}
% From \eqref{eq:prod_pdf}, with $U\defeq Z/s$,
% \[
% f_U(u)=\frac{1}{\pi\sqrt{1-\rho^2}}\exp\!\left(\frac{\rho u}{1-\rho^2}\right)
% K_0\!\left(\frac{|u|}{1-\rho^2}\right),
% \]
% and
% \[
% h(U)=-\int_{\R} f_U(u)\log f_U(u)\dd u,
% \qquad
% N(U)=\frac{1}{2\pi e}\exp(2h(U)).
% \]
% This can be evaluated numerically for any $\rho\in(-1,1)$, and then used in the SLB benchmarks
% \eqref{eq:LB_Z}--\eqref{eq:LB_S_explicit}.
\newpage
\bibliographystyle{plain} %% Using plain natbib style
\bibliography{neurips_2025.bib}

@article{lloyd1982least,
  title={Least squares quantization in PCM},
  author={Lloyd, Stuart},
  journal={IEEE transactions on information theory},
  volume={28},
  number={2},
  pages={129--137},
  year={1982},
  publisher={IEEE}
}

@article{dettmers2023qlora,
  title={Qlora: Efficient finetuning of quantized llms},
  author={Dettmers, Tim and Pagnoni, Artidoro and Holtzman, Ari and Zettlemoyer, Luke},
  journal={Advances in neural information processing systems},
  volume={36},
  pages={10088--10115},
  year={2023}
}

@misc{dettmers2022llmint88bitmatrixmultiplication,
      title={LLM.int8(): 8-bit Matrix Multiplication for Transformers at Scale}, 
      author={Tim Dettmers and Mike Lewis and Younes Belkada and Luke Zettlemoyer},
      year={2022},
      eprint={2208.07339},
      archivePrefix={arXiv},
      primaryClass={cs.LG},
      url={https://arxiv.org/abs/2208.07339}, 
}

@misc{micikevicius2022fp8formatsdeeplearning,
      title={FP8 Formats for Deep Learning}, 
      author={Paulius Micikevicius and Dusan Stosic and Neil Burgess and Marius Cornea and Pradeep Dubey and Richard Grisenthwaite and Sangwon Ha and Alexander Heinecke and Patrick Judd and John Kamalu and Naveen Mellempudi and Stuart Oberman and Mohammad Shoeybi and Michael Siu and Hao Wu},
      year={2022},
      eprint={2209.05433},
      archivePrefix={arXiv},
      primaryClass={cs.LG},
      url={https://arxiv.org/abs/2209.05433}, 
}

@misc{merity2016pointersentinelmixturemodels,
      title={Pointer Sentinel Mixture Models}, 
      author={Stephen Merity and Caiming Xiong and James Bradbury and Richard Socher},
      year={2016},
      eprint={1609.07843},
      archivePrefix={arXiv},
      primaryClass={cs.CL},
      url={https://arxiv.org/abs/1609.07843}, 
}

@article{radford2019language,
  title={Language models are unsupervised multitask learners},
  author={Radford, Alec and Wu, Jeffrey and Child, Rewon and Luan, David and Amodei, Dario and Sutskever, Ilya and others},
  journal={OpenAI blog},
  volume={1},
  number={8},
  pages={9},
  year={2019}
}

@misc{yang2025qwen3technicalreport,
      title={Qwen3 Technical Report}, 
      author={An Yang and Anfeng Li and Baosong Yang and Beichen Zhang and Binyuan Hui and Bo Zheng and Bowen Yu and Chang Gao and Chengen Huang and Chenxu Lv and Chujie Zheng and Dayiheng Liu and Fan Zhou and Fei Huang and Feng Hu and Hao Ge and Haoran Wei and Huan Lin and Jialong Tang and Jian Yang and Jianhong Tu and Jianwei Zhang and Jianxin Yang and Jiaxi Yang and Jing Zhou and Jingren Zhou and Junyang Lin and Kai Dang and Keqin Bao and Kexin Yang and Le Yu and Lianghao Deng and Mei Li and Mingfeng Xue and Mingze Li and Pei Zhang and Peng Wang and Qin Zhu and Rui Men and Ruize Gao and Shixuan Liu and Shuang Luo and Tianhao Li and Tianyi Tang and Wenbiao Yin and Xingzhang Ren and Xinyu Wang and Xinyu Zhang and Xuancheng Ren and Yang Fan and Yang Su and Yichang Zhang and Yinger Zhang and Yu Wan and Yuqiong Liu and Zekun Wang and Zeyu Cui and Zhenru Zhang and Zhipeng Zhou and Zihan Qiu},
      year={2025},
      eprint={2505.09388},
      archivePrefix={arXiv},
      primaryClass={cs.CL},
      url={https://arxiv.org/abs/2505.09388}, 
}

@article{bennett1948spectra,
  author = {Bennett, W. R.},
  title = {Spectra of Quantized Signals},
  journal = {Bell System Technical Journal},
  volume = {27},
  number = {3},
  pages = {446--472},
  year = {1948},
  doi = {10.1002/j.1538-7305.1948.tb01340.x}
}

@book{gersho2012vector,
  title={Vector quantization and signal compression},
  author={Gersho, Allen and Gray, Robert M},
  volume={159},
  year={2012},
  publisher={Springer Science \& Business Media}
}

@misc{nvidia2025nvfp4,
  author = {{NVIDIA Technical Blog}},
  title = {Introducing NVFP4 for Efficient and Accurate Low-Precision Inference},
  year = {2025},
  month = jun,
  day = {24},
  note = {NVIDIA Technical Blog}
}

@article{ordentlichpolyanskiy2025optimal,
  title={Optimal quantization for matrix multiplication},
  author={Ordentlich, Or and Polyanskiy, Yury},
  journal={IEEE Transactions on Information Theory},
  year={2025},
  publisher={IEEE}
}

@book{jayant1984digital,
  title={Digital coding of waveforms: principles and applications to speech and video},
  author={Jayant, Nuggehally S and Noll, Peter},
  volume={2},
  year={1984},
  publisher={Prentice-Hall Englewood Cliffs, NJ}
}
\end{document}